\documentclass[11pt]{article}

\usepackage[utf8]{inputenc}
\usepackage{macros}
\usepackage{mathtools}
\usepackage{mathptmx}
\usepackage{tocloft}

\newcommand\drawsymbol[3]{\node at (#1,#2) {\tiny #3};}
\def\xspacing{0.3}
\def\xmetaspacing{2.4}
\newcommand\draweightsymbols[3]{
  \def\tempx{#1}
  \def\tempy{#2}
  \draweightsymbolscontinued#3
}

\def\draweightsymbolscontinued#1,#2,#3,#4,#5,#6,#7,#8{
    \drawsymbol{\tempx+\xspacing*0}{\tempy}{#1}
    \drawsymbol{\tempx+\xspacing*1}{\tempy}{#2}
    \drawsymbol{\tempx+\xspacing*2}{\tempy}{#3}
    \drawsymbol{\tempx+\xspacing*3}{\tempy}{#4}
    \drawsymbol{\tempx+\xspacing*4}{\tempy}{#5}
    \drawsymbol{\tempx+\xspacing*5}{\tempy}{#6}
    \drawsymbol{\tempx+\xspacing*6}{\tempy}{#7}
    \drawsymbol{\tempx+\xspacing*7}{\tempy}{#8}
}

\newcommand\drawmetasymbol[3]{
  \ifx1#3
   \draweightsymbols{#1*\xmetaspacing}{#2}{1,0,1,1,0,1,1,0}
  \fi
  \ifx2#3
    \draweightsymbols{#1*\xmetaspacing}{#2}{0,0,0,0,1,0,0,1}
  \fi
  \ifx3#3
    \draweightsymbols{#1*\xmetaspacing}{#2}{1,0,1,1,1,0,1,0}
  \fi
  \ifx4#3
   \draweightsymbols{#1*\xmetaspacing}{#2}{0,1,1,0,0,1,1,0}
  \fi
  \ifx5#3
    \draweightsymbols{#1*\xmetaspacing}{#2}{1,1,0,1,0,0,0,1}
  \fi
  \ifx6#3
    \draweightsymbols{#1*\xmetaspacing}{#2}{0,1,1,0,0,0,1,1}
  \fi
  \ifx7#3
   \draweightsymbols{#1*\xmetaspacing}{#2}{1,1,0,1,0,0,0,0}
  \fi
  \ifx8#3
   \draweightsymbols{#1*\xmetaspacing}{#2}{1,0,1,0,1,0,1,1}
  \fi
  \ifx0#2
    \node at (#1*\xmetaspacing+3.5*\xspacing,0) {};
  \fi
  \ifx1#2
    \node at (#1*\xmetaspacing+3.5*\xspacing,1.4) {};
  \fi
}

\def\parsepos#1#2,#3{\def #1{#2*\xmetaspacing + #3*\xspacing}}
\newcommand\drawmatch[2]{
   \parsepos\posA#1
   \parsepos\posB#2
   \draw (\posA,0.1) -- (\posB,0.9);
}
\newcommand\drawmatchlight[2]{
   \parsepos\posA#1
   \parsepos\posB#2
   \draw[dashed] (\posA,0.1) -- (\posB,0.9);
}
\newcommand\drawmatchgreen[2]{
   \parsepos\posA#1
   \parsepos\posB#2
   \draw[line width=2, color=black!60!green] (\posA,0.1) -- (\posB,0.9);
}

\def\ind{0.11}
\def\drawpart#1,#2--#3:#4,#5--#6{
 \fill[top color=blue!20, bottom color=red!20] 
    (#1*\xmetaspacing+#2*\xspacing-\ind,-\ind) -- 
    (#1*\xmetaspacing+#2*\xspacing-\ind,+\ind) --
    (#4*\xmetaspacing+#5*\xspacing-\ind,1-\ind) -- 
    (#4*\xmetaspacing+#5*\xspacing-\ind,1+\ind) --
    (#4*\xmetaspacing+#6*\xspacing+\ind,1+\ind) -- 
    (#4*\xmetaspacing+#6*\xspacing+\ind,1-\ind) --
    (#1*\xmetaspacing+#3*\xspacing+\ind,\ind) -- 
    (#1*\xmetaspacing+#3*\xspacing+\ind,-\ind) -- 
    cycle;
}

\def\ind{0.11}
\newcommand\drawpartg[8]{
 \fill[green!20] 
    (#1*\xmetaspacing+#2*\xspacing-\ind,-\ind) -- 
    (#1*\xmetaspacing+#2*\xspacing-\ind,+\ind) --
    (#5*\xmetaspacing+#6*\xspacing-\ind,1-\ind) -- 
    (#5*\xmetaspacing+#6*\xspacing-\ind,1+\ind) --
    (#7*\xmetaspacing+#8*\xspacing+\ind,1+\ind) -- 
    (#7*\xmetaspacing+#8*\xspacing+\ind,1-\ind) --
    (#3*\xmetaspacing+#4*\xspacing+\ind,\ind) -- 
    (#3*\xmetaspacing+#4*\xspacing+\ind,-\ind) -- 
    cycle;
}

\def\ind{0.11}
\newcommand\drawpartgg[8]{
 \fill[green] 
    (#1*\xmetaspacing+#2*\xspacing-\ind,-\ind) -- 
    (#1*\xmetaspacing+#2*\xspacing-\ind,+\ind) --
    (#5*\xmetaspacing+#6*\xspacing-\ind,1-\ind) -- 
    (#5*\xmetaspacing+#6*\xspacing-\ind,1+\ind) --
    (#7*\xmetaspacing+#8*\xspacing+\ind,1+\ind) -- 
    (#7*\xmetaspacing+#8*\xspacing+\ind,1-\ind) --
    (#3*\xmetaspacing+#4*\xspacing+\ind,\ind) -- 
    (#3*\xmetaspacing+#4*\xspacing+\ind,-\ind) -- 
    cycle;
}

\def\ind{0.11}
\newcommand\drawpartby[8]{
 \fill[top color=blue!30, bottom color=yellow!30] 
    (#1*\xmetaspacing+#2*\xspacing-\ind,-\ind) -- 
    (#1*\xmetaspacing+#2*\xspacing-\ind,+\ind) --
    (#5*\xmetaspacing+#6*\xspacing-\ind,1-\ind) -- 
    (#5*\xmetaspacing+#6*\xspacing-\ind,1+\ind) --
    (#7*\xmetaspacing+#8*\xspacing+\ind,1+\ind) -- 
    (#7*\xmetaspacing+#8*\xspacing+\ind,1-\ind) --
    (#3*\xmetaspacing+#4*\xspacing+\ind,\ind) -- 
    (#3*\xmetaspacing+#4*\xspacing+\ind,-\ind) -- 
    cycle;
}

\def\ind{0.11}
\newcommand\drawpartbyby[8]{
 \fill[top color=blue!60, bottom color=yellow!60] 
    (#1*\xmetaspacing+#2*\xspacing-\ind,-\ind) -- 
    (#1*\xmetaspacing+#2*\xspacing-\ind,+\ind) --
    (#5*\xmetaspacing+#6*\xspacing-\ind,1-\ind) -- 
    (#5*\xmetaspacing+#6*\xspacing-\ind,1+\ind) --
    (#7*\xmetaspacing+#8*\xspacing+\ind,1+\ind) -- 
    (#7*\xmetaspacing+#8*\xspacing+\ind,1-\ind) --
    (#3*\xmetaspacing+#4*\xspacing+\ind,\ind) -- 
    (#3*\xmetaspacing+#4*\xspacing+\ind,-\ind) -- 
    cycle;
}

\newcommand\drawpartbynoind[8]{
 \fill[top color=blue!30, bottom color=yellow!30] 
    (#1*\xmetaspacing+#2*\xspacing,0) -- 
    (#1*\xmetaspacing+#2*\xspacing,0) --
    (#5*\xmetaspacing+#6*\xspacing,1) -- 
    (#5*\xmetaspacing+#6*\xspacing,1) --
    (#7*\xmetaspacing+#8*\xspacing,1) -- 
    (#7*\xmetaspacing+#8*\xspacing,1) --
    (#3*\xmetaspacing+#4*\xspacing,0) -- 
    (#3*\xmetaspacing+#4*\xspacing,0) -- 
    cycle;
}

\def\ind{0.11}
\newcommand\drawpartyb[8]{
 \fill[top color=yellow!30, bottom color=blue!30] 
    (#1*\xmetaspacing+#2*\xspacing-\ind,-\ind) -- 
    (#1*\xmetaspacing+#2*\xspacing-\ind,+\ind) --
    (#5*\xmetaspacing+#6*\xspacing-\ind,1-\ind) -- 
    (#5*\xmetaspacing+#6*\xspacing-\ind,1+\ind) --
    (#7*\xmetaspacing+#8*\xspacing+\ind,1+\ind) -- 
    (#7*\xmetaspacing+#8*\xspacing+\ind,1-\ind) --
    (#3*\xmetaspacing+#4*\xspacing+\ind,\ind) -- 
    (#3*\xmetaspacing+#4*\xspacing+\ind,-\ind) -- 
    cycle;
}

\def\ind{0.11}
\newcommand\drawpartybyb[8]{
 \fill[top color=yellow!60, bottom color=blue!60] 
    (#1*\xmetaspacing+#2*\xspacing-\ind,-\ind) -- 
    (#1*\xmetaspacing+#2*\xspacing-\ind,+\ind) --
    (#5*\xmetaspacing+#6*\xspacing-\ind,1-\ind) -- 
    (#5*\xmetaspacing+#6*\xspacing-\ind,1+\ind) --
    (#7*\xmetaspacing+#8*\xspacing+\ind,1+\ind) -- 
    (#7*\xmetaspacing+#8*\xspacing+\ind,1-\ind) --
    (#3*\xmetaspacing+#4*\xspacing+\ind,\ind) -- 
    (#3*\xmetaspacing+#4*\xspacing+\ind,-\ind) -- 
    cycle;
}

\newcommand\drawpartybnoind[8]{
 \fill[top color=yellow!30, bottom color=blue!30] 
    (#1*\xmetaspacing+#2*\xspacing,0) -- 
    (#1*\xmetaspacing+#2*\xspacing,0) --
    (#5*\xmetaspacing+#6*\xspacing,1) -- 
    (#5*\xmetaspacing+#6*\xspacing,1) --
    (#7*\xmetaspacing+#8*\xspacing,1) -- 
    (#7*\xmetaspacing+#8*\xspacing,1) --
    (#3*\xmetaspacing+#4*\xspacing,0) -- 
    (#3*\xmetaspacing+#4*\xspacing,0) -- 
    cycle;
}
\newcommand\drawpartgraynoind[8]{
 \fill[top color=gray!30, bottom color=gray!30] 
    (#1*\xmetaspacing+#2*\xspacing,0) -- 
    (#1*\xmetaspacing+#2*\xspacing,0) --
    (#5*\xmetaspacing+#6*\xspacing,1) -- 
    (#5*\xmetaspacing+#6*\xspacing,1) --
    (#7*\xmetaspacing+#8*\xspacing,1) -- 
    (#7*\xmetaspacing+#8*\xspacing,1) --
    (#3*\xmetaspacing+#4*\xspacing,0) -- 
    (#3*\xmetaspacing+#4*\xspacing,0) -- 
    cycle;
}

\def\subseqpic{
\drawmetasymbol{0}{1}{1}
\drawmetasymbol{1}{1}{2}
\drawmetasymbol{2}{1}{3}
\drawmetasymbol{3}{1}{4}
\drawmetasymbol{0}{0}{5}
\drawmetasymbol{1}{0}{6}
\drawmetasymbol{2}{0}{7}
\drawmetasymbol{3}{0}{8}
}

\newcommand\toplabel[2]{
\node[inner sep=0,outer sep=0, label={[anchor=south]\tiny$#2$}] at (#1*\xspacing,1) {};
\draw (#1*\xspacing,1.1)--(#1*\xspacing,1.0);
}
\newcommand\bottomlabel[2]{
\node[inner sep=0,outer sep=0, label={[anchor=north]\tiny$#2$}] at (#1*\xspacing,0) {};
\draw (#1*\xspacing,-0.1)--(#1*\xspacing,0);
}
\newcommand\toplabelshift[2]{
\node[inner sep=0,outer sep=0, label={[anchor=south]\tiny$#2$}] at (#1,1) {};
\draw (#1,1.1)--(#1,1);
}
\newcommand\bottomlabelshift[2]{
\node[inner sep=0,outer sep=0, label={[anchor=north]\tiny$#2$}] at (#1,0) {};
\draw (#1,-0.1)--(#1,0);
}

\title{The zero-rate threshold for adversarial bit-deletions is less than $\tfrac{1}{2}$\footnote{An extended abstract of this work was presented at the 2021 Foundations of Computer Science (FOCS) conference. This is an expanded version including all the proofs.}}
\author{Venkatesan Guruswami\thanks{Computer Science Department, Carnegie Mellon University. venkatg@cs.cmu.edu. Research supported in part by NSF grants CCF-1814603 and a Simons Investigator Award.} \and Xiaoyu He\thanks{Department of Mathematics, Princeton University. xiaoyuh@princeton.edu. Research supported by NSF Grant DGE-1656518.} \and Ray Li\thanks{Department of Computer Science, Stanford University. rayyli@cs.stanford.edu. Research supported by NSF Grants DGE-1656518, CCF-1814629, and by Jacob Fox's Packard Fellowship.}}
\date{}

\global\long\def\I{\mathcal{I}}%

\global\long\def\rev{\textnormal{rev}}%

\global\long\def\floor#1{\lfloor#1\rfloor}%

\global\long\def\ceil#1{\lceil#1\rceil}%

\global\long\def\eps{\varepsilon}%

\global\long\def\Drop{\textnormal{Trim}}%

\global\long\def\ZZ{\mathbb{Z}}%

\global\long\def\indicator{\mathbbm{1}}%
\def\thr{\text{thr}}
\def\del{\text{del}}

\begin{document}

\maketitle
\vspace{-0.7cm}
\begin{abstract}
  We prove that there exists an absolute constant $\delta>0$ such that any binary code $C\subset\{0,1\}^N$ tolerating $(1/2-\delta)N$ adversarial deletions must satisfy $|C|\le 2^{\poly\log N}$ and thus have rate asymptotically approaching 0.
  This is the first constant fraction improvement over the trivial bound that codes tolerating $N/2$ adversarial deletions must have rate going to 0 asymptotically.
  Equivalently, we show that there exists absolute constants $A$ and $\delta>0$ such that any set $C\subset\{0,1\}^N$ of $2^{\log^A N}$ binary strings must contain two strings $c$ and $c'$ whose longest common subsequence has length at least $(1/2+\delta)N$.
  As an immediate corollary, we show that $q$-ary codes tolerating a fraction $1-(1+2\delta)/q$ of adversarial deletions must also have rate approaching 0.
  \smallskip
  
   Our techniques include string regularity arguments and a structural lemma that classifies binary strings by their oscillation patterns.  Leveraging these tools, we find in any large code two strings with similar oscillation patterns, which is exploited to find a long common subsequence.

\end{abstract}
\vspace{-0.3cm}
{\footnotesize
\tableofcontents}

\thispagestyle{empty}

\newpage

\section{Introduction}

This work considers the limits of reliable communication against an adversarial deletion channel. Suppose we are to transmit $N$ bits on a channel that can adversarially delete a fraction $p$ of the bits, leaving the receiver with a subsequence of length $(1-p)N$. Crucially, the receiver does not know the location of the deleted bits. We would like to achieve zero-error communication over such an adversarial deletion channel. 
To do so, we restrict the sequence of transmitted bits to a subset $C \subset \{0,1\}^N$ so that every $x \in C$ can be unambiguously identified from an arbitrary subsequence of $x$ of length $(1-p)N$. It is easy to see that this property is equivalent to the property that for every two distinct codewords $x,y \in C$, the length of their longest common subsequence, denoted $\LCS(x,y)$, is less than $(1-p)N$. 

Defining $\LCS(C)$ to be the largest value of $\LCS(x,y)$ over all distinct pairs $x,y \in C$, we therefore call a subset $C \subseteq \{0,1\}^N$ a \emph{$p$-deletion correcting code} if $\LCS(C) < (1-p)N$.
The rate of such a code, defined as $R(C) = (\log_2 |C|)/N$, measures the average number of information bits communicated per transmitted codeword bit. If there exists a family of such codes $C$ whose rates are bounded away from zero as $N \to \infty$, we say it is possible to achieve a non-vanishing rate of information communication. 

For any noise model of interest, a fundamental question is to understand its ``capacity,'' i.e., the information-theoretically optimal trade-off between rate and noise level. 
This is often very challenging and open in most cases of interest in zero-error information theory. A less demanding, and still very fundamental, goal is to understand the threshold noise level below which it is possible to communicate with non-vanishing rate.  For example, it is well-known that for a channel that flips an adversarially chosen set of at most $pN$ bits, the threshold value for the error-fraction $p$ equals $\tfrac14$. 
On the other hand, for the adversarial deletions channel, this fundamental question remains unsolved:
\begin{quote}
    \emph{What is the largest fraction of deletions $p\in(0,1)$ for which it is possible to achieve zero-error communication with non-vanishing information rate?}
\end{quote}
Formally, define the \emph{zero-rate threshold} of adversarial bit-deletions to be 
\begin{align}
    p^{\thr}_{\del} & := \sup \{ p \mid \text{$\exists \ \alpha_p > 0$ s.t for infinitely many $N$ there is a subset $C\subset \{0,1\}^N$ with} \nonumber \\ &  \qquad\qquad\quad \LCS(C) < (1-p)N \text{ and } |C|\ge 2^{\alpha_p N} \} \ . 
\label{eq:thr-del-defn}
\end{align}

The above question is then: What is the value of $p^{\thr}_{\del}$? We have a trivial upper bound $p^{\thr}_{\del} \le 1/2$. Indeed, among any three strings $x,y,z \in \{0,1\}^N$, there must be two with the same majority bit, and thus a common subsequence (of all $0$'s or all $1$'s) of length at least $N/2$. Thus any $\sfrac{1}{2}$-deletion correcting code $C \subset \{0,1\}^N$ satisfies $|C| \le 2$.

The value of $p^{\thr}_{\del}$ remains unknown. Even more starkly, as simplistic as the above argument is, it was not previously known if $p^{\thr}_{\del}$ is strictly bounded away from $\sfrac{1}{2}$, or whether there are in fact codes of non-vanishing rate for correcting a $(\sfrac12-\delta)$ fraction of deletions for any desired $\delta > 0$. This tantalizing question was implicit in early works on deletion codes, particularly in \cite{U67}, which gave bounds on the achievable tradeoffs between rate and deletion fraction, and was explicitly raised in \cite{KMTU11}. Since then, this question has been mentioned in several works, including the work of Bukh and Ma~\cite{bukh_ma} which showed that an upper bound of $\tfrac12-\tfrac{1}{\poly \log  N}$ on the correctable deletion fraction, and many recent works on deletion code constructions such as \cite{Wang15,GL-isit16, GW17,BGH17,GL-oblivious,GHS-stoc20}, other works on coding theory \cite{ZBJ20}, as well as the recent surveys~\cite{CR-survey,HS-survey}. 
In the other direction, the best known lower bound $p^{\thr}_{\del} > \sqrt{2} - 1$ is due to \cite{BGH17}, who constructed explicit binary codes of non-vanishing rate to correct a fraction of deletions approaching $\sqrt{2}-1$ (see \cite{SZ99, KMTU11} for prior constructions).

\subsection{Our results}
In this work, we prove the first nontrivial upper bound on $p^{\thr}_{\del}$.
\begin{theorem}
  \label{thm:main-intro}
  There exists an absolute constant $\delta_0 > 0$ such that $p^{\thr}_{\del} \le \tfrac{1}{2}-\delta_0$.  More concretely, there exist absolute constants $A, \delta_0>0$ such that for all large enough $N$, any binary code $C\subset\{0,1\}^N$ tolerating $(\tfrac12-\delta_0)N$ adversarial deletions must satisfy $|C|\le 2^{(\log N)^A}$. 
\end{theorem}

We show the above in the contrapositive form---in any code $C\subset \{0,1\}^N$ of quasi-polynomial size, we find two codewords $s,t \in C$ with $\LCS(s,t) > (\sfrac12+\delta_0)N$. Furthermore, our proofs obtain an explicit value $\delta_0 \ge 10^{-40}$; we made no attempts to optimize the value of $\delta_0$ but regardless it is very small. Despite the tiny improvement over the trivial $\tfrac12$ bound, the above result is qualitatively significant because it rules out correcting deletion fractions arbitrarily close to $\tfrac12$.

In Section~\ref{sec:overview}, we give an overview of the proof and an outline of this paper. In the remainder of this introduction we survey some generalizations of Theorem~\ref{thm:main-intro} and connections to other problems in coding theory.

\medskip\noindent\textbf{Non-binary alphabets.}
We generalize Theorem~\ref{thm:main-intro} to alphabets of larger size. Let us denote the quantity analogous to \eqref{eq:thr-del-defn} for any fixed alphabet size $q\ge 2$, namely the zero-rate threshold for $q$-ary deletion codes, by $p^{\thr}_{\del}(q)$. The trivial upper bound is $p^{\thr}_{\del}(q) \le 1-1/q$; this corresponds to finding a common sequence of at least $N/q$ repeated $i$'s between two strings that share the same most frequent symbol $i\in \{0,1,\dots,q-1\}$, in any code of size bigger than $q$. Just as in the binary case, no improvement over this trivial bound was previously known. 

For any code $C \subseteq \{0,1,\dots, q-1\}^N$ over an alphabet of size $q > 2$, we may pick some two symbols $i,j$ and a set $C_{i,j}\subseteq C$ of at least $|C|/q^2$ strings whose two most frequent symbols are $i$ and $j$. We can then obtain a binary code $C'\subseteq \{i,j\}^{2N/q}$ by restricting each element of $C_{i,j}$ to a substring of length\footnote{To be pedantic we should write $\lceil 2N/q\rceil$ here. Henceforth we omit floors and ceilings where they are not essential.} $2N/q$ consisting only of $i$'s and $j$'s. Applying the contrapositive form of Theorem~\ref{thm:main-intro} to $C'$, we see that as long as $|C'| > 2^{(\log n)^A}$, some two strings in $C'$ have a common subsequence with length at least $(\tfrac{1}{2}+\delta_0) \tfrac{2N}{q}$. We have thus shown the following theorem as an immediate corollary of Theorem~\ref{thm:main-intro}.
\begin{theorem}
 Fix an integer $q 
 \ge 2$. Then 
\[ p^{\thr}_{\del}(q) \le 1 - \tfrac{1+2\delta_0}{q}  < 1-\tfrac1q \ , \]
where $\delta_0>0$ is the positive constant promised in Theorem~\ref{thm:main-intro}.
\end{theorem}

We note that for the simpler model of erasures where the location of missing symbols \emph{are} known to the decoder, the zero-rate threshold equals $1-1/q$.\footnote{The erasure fraction correctable by a code is exactly governed by its relative (Hamming) distance. The Plotkin bound shows that the rate must be vanishing for relative Hamming distance $1-1/q$. We know the existence and even explicit constructions of codes of non-vanishing rate and relative Hamming distance $1-1/q-\eps$ for any $\eps > 0$.}
Thus our results also show a formal separation between the zero-rate threshold for the models of erasures and deletions, for any fixed alphabet.

 In the list-decoding model with list-size $L$ for deletion fraction $p$, there can be up to $L$ codewords that contain the decoder's (arbitrary) input sequence $y \in \{0,1\}^{(1-p)N}$. 
The zero-rate threshold for \emph{list-decoding} from deletions, as the list-size $L \to \infty$, is known to equal $1-1/q$~\cite{GW17}. Thus our result also demonstrates that list-decoding is provably more powerful in terms of the deletion fractions that can be handled with non-vanishing rate. 

\subsection{Related works} 

\smallskip\noindent\textbf{Performance of random codes.} An ubiquitous approach to establish strong, and in many cases the best known, \emph{possibility} results in coding theory is to analyze random codes of certain rates, i.e. codes generated by sampling i.i.d. random elements of $\{0,1\}^N$. These results typically also identify the precise performance threshold of random codes~\cite{GMRSW21}. For the case of binary deletion codes, however, the performance of random codes itself is hard to analyze, and we do not rigorously know a tight estimate of the expected length $\gamma N$ of the longest common subsequence of two random $N$-bit strings ($\gamma$ is called the Chv{\'a}tal-Sankoff constant \cite{CS75}). The known bounds on this expectation $\gamma N$ \cite{L09}, together with standard probabilistic arguments, imply that with high probability, random codes can tolerate a deletion fraction at least $0.17$, but also at most $0.22$ \cite{KMTU11}.  For codes over alphabet size $q$, random codes can correct a deletion fraction approaching $1-2/\sqrt{q}$ for large $q$~\cite{KLM}.

As mentioned earlier, we now have constructions of binary codes that can correct a deletion fraction $0.414$~\cite{BGH17}, which is substantially better than random codes. This raised the possibility that perhaps there might be binary codes of non-vanishing rate capable of correcting a deletion fraction all the way up to the trivial limit of $\sfrac12$, which we refute in this work.

\medskip\noindent\textbf{Trade-offs for correcting $N/2 - N^{1-\theta}$ deletions.}
In terms of previously known limitations of deletion codes,
Bukh and Ma~\cite{bukh_ma} showed that for each fixed $r$ and large enough $N$ (specifically, at least $r^{O(r)}$), every set $C \subseteq \{0,1\}^N$ of size $r+4$ satisfies 
\begin{equation}
    \label{eq:bukh_ma}
\LCS(C) \ge \tfrac{N}{2} + \Omega(r^{-9}) N^{1-1/r} \ . 
\end{equation}
Choosing $r$ appropriately, the result \eqref{eq:bukh_ma} implies that there exist absolute constants $b,c$ such that 
every code $C \subseteq \{0,1\}^N$ 
with $\LCS(C) < \tfrac{N}{2} + \tfrac{N}{(\log N)^b}$ has size $|C| \le c \tfrac{\log N}{\log \log N}$. 
Bukh and Ma demonstrated that this $N^{1-1/r}$ advantage in \eqref{eq:bukh_ma} is asymptotically sharp for each fixed $r$, by exhibiting a set $\mathcal{W}$ of $(r+4)$ $N$-bit strings with $\LCS(\mathcal{W}) \le \tfrac{N}{2} + O(N^{1-1/r})$. Interestingly, this set $\mathcal{W}$ played a crucial role in the developments on \emph{constructions} of codes to correct a large fraction of deletions in \cite{BGH17,GL-oblivious}, as well as codes achieving the zero-rate threshold for list decoding from insertions and deletions in \cite{GHS-stoc20}. A suitable modification of this Bukh-Ma code $\mathcal{W}$ also drives the best known $0.414$-deletion correcting codes of \cite{BGH17}.

\medskip\noindent\textbf{Twins and regularity techniques.} One of the ideas used in this work is a new regularity-type result about strings. Szemer\'edi's regularity lemma and its variants are ubiquitous in extremal and additive combinatorics, but applying these ideas to coding theory is a relatively recent development. The first example of such a result was proved by~\cite{APP13}, and their regularity lemma roughly shows that every long string can be partitioned into a constant number of consecutive substrings, each of which is regular (a regular string is one in which the one-density in any long consecutive substring is close to the one-density in the whole string). They used this regularity lemma to prove that every string of length $N$ contains two disjoint copies of some length $(1/2 - o(1))N$ subsequence (so-called ``twins''). Given the similarity between finding twins in a single string and finding long common sequences between different strings, it should not come as a surprise that these techniques are useful here as well. The main difference in our approach is that we require a stronger regularity condition, which is that every substring not only has the same one-density but has similar ``oscillation statistics'' at many scales with the parent string.

\subsection{Deletion correction in related models}

To offer some wider context, we now discuss some results related to the broader study of codes for deletions and synchronization errors under various channel assumptions.

\medskip\noindent\textbf{Non-adversarial models.} Our work focuses on the adversarial model, where an arbitrary subset of $p$ fraction of the codeword bits, can be deleted. There is a rich body of work on the binary deletion channel where each codeword bit is deleted i.i.d with probability $p$. In this case, it is known that one can have positive rate codes that ensure vanishing miscommunication probability even for $p$ approaching $1$ (so the zero-rate threshold equals $1$). 
The interested reader can find more information about codes for the deletion channel in the surveys~\cite{Mitzenmacher-survey,CR-survey}.

One can consider models that are intermediate in power between i.i.d random and adversarial channels. For instance, in the \emph{oblivious} model the deletion pattern can be chosen arbitrarily, but without knowledge of the codeword. In this case, too, the zero-rate threshold is $1$, as for any $p < 1$, Guruswami and Li~\cite{GL-oblivious} showed the existence of codes that ensured that for every pattern of $p$-fraction deletions \emph{most} codewords are communicated correctly.\footnote{This average-case criterion to achieve decoding success for {most}, as opposed to {all}, codewords is necessary, as otherwise tackling the oblivious model becomes as hard as tackling the adversarial model. Alternatively, one can allow a stochastic encoder, and ensure high probability of successful transmission of each message when averaged over the choice of its random encoding. See, for example, \cite[Appendix A]{GL-oblivious}.}
Their work also considered the \emph{online} model, where the decision to delete the $i$-th bit must be made based only on the first $i$ bits of the codeword. They showed that the zero-rate threshold for this model (again, for the average-error criterion of ensuring most codewords are communicated correctly for any deletion pattern) equals $\sfrac12$ if and only if $p^{\thr}_{\del} =\sfrac12$. By virtue of Theorem~\ref{thm:main-intro}, this implies that the zero-rate threshold for the online model is also bounded away from $\sfrac12$. 

\medskip\noindent\textbf{Large alphabets.} We focused on codes over the binary and fixed small alphabets in this work. This is in fact the most challenging setting for deletion codes. Indeed, if the code alphabet is allowed to grow with $N$, then one can include the index $i$ along with the $i$-th codeword symbol, effectively reducing the deletion model to the much simpler erasure model, where Reed-Solomon codes give a simple, optimal solution. For alphabets that are large, but still independent of $N$, a natural greedy strategy shows the existence of
codes of rate $(1-p-\eps)$ capable of correcting a fraction $p$ of deletions, over an alphabet of size $\exp(O(1/\eps))$~\cite{GW17}. In particular, the zero-rate threshold approaches $1$. Also, $1-p$ is a trivial upper bound on the possible rate, even for the simpler model of $p$ fraction of erasures. Explicit constructions of $p$-deletion correcting of rate approaching $1-p$ over an alphabet size independent of $N$ were given in \cite{HS17} based on synchronization strings, which is a very elegant tool that has since found several other applications (see the survey~\cite{HS-survey}).

\medskip\noindent\textbf{Insertions and deletions.} Another form of errors that affect the synchronization between sender and receiver are insertions of symbols. It is well known (since \cite{Lev66}) that a code $C$ with $\LCS(C) < (1-p)N$ can tolerate any combination of a total of $pN$ insertions and deletions (insdel errors). Thus allowing insertions as well does not change the combinatorial aspects of the underlying coding problem, as it is governed by the LCS. However, for efficient algorithms for insdel errors are not implied by deletion correction algorithms, and have to be reworked~\cite{GL-isit16}.

In the model of list-decoding, even the combinatorial aspects are more nuanced in the presence of insertions. The trade-off between the combinations of fractions of insertions and deletions that governs the zero-rate region exhibits an interesting piece-wise linear behavior~\cite{GHS-stoc20}. (As mentioned earlier, the zero-rate threshold for list-decoding $q$-ary codes from deletions alone equals $1-1/q$.)

\medskip\noindent\textbf{Low-deletions regime.} This work focused on the largest deletion fraction that can be corrected with non-vanishing rate. At the opposite end of the spectrum are codes to correct a deletion fraction $p \to 0$. In this case, the optimal rate behaves as $1-O(p \log (1/p))$ and we also know explicit codes with such rate $1-O(p \log^2 (1/p))$ and efficient deletion-correction algorithms~\cite{CJLW18,Haeupler19}.
There has also be an active line of recent code constructions, triggered by \cite{BGZ18}, for correction of a fixed \emph{number} $k$ of deletions with redundancy at most $c_k \log N$. We now have codes with the optimal (up to constant factors) redundancy of $O(k \log N)$~\cite{CJLW18,Haeupler19,SimaB21,SimaGB20}.

\section{Proof overview}
\label{sec:overview}
In this section, we give a high-level overview of the proof of Theorem~\ref{thm:main-intro}, as well as the organization of the rest of this paper. Let $C\subseteq \{0,1\}^N$ be a binary code of size $2^{(\log N)^A}$ for some large constant $A$, and our goal will be to find two elements $s,t\in C$ for which $\LCS(s,t) \ge (1/2 + \delta_0)N$. The proof breaks down into five conceptually independent parts that roughly correspond to the Sections~\ref{sec:structure} through~\ref{sec:wrapup}. It will be natural to explain these parts here in roughly reverse order, starting with Section~\ref{sec:wrapup}.

\medskip\noindent\textbf{Pigeonholing by ``statistics''.}
The only place where we use the size of $C$ is in the final Section~\ref{sec:wrapup}, which wraps up the proof of Theorem~\ref{thm:main-intro}. We need $C$ to be large enough to find by the pigeonhole principle two elements $s,t \in C$ with similar ``macroscopic statistics.'' That is, partition $[N]$ into long subintervals of length $N/\poly\log N$. We pick $s$ and $t$ to have the same number of ones in each of these intervals, and also to share some other statistics that characterize the ``frequencies at which they oscillate.'' As there are only $O(\poly\log N)$ long intervals and the statistics in question take on only $2^{\poly\log N}$ possible values on each interval, $C$ is large enough to guarantee the existence of two $s$ and $t$ sharing identical statistics on all such intervals. The remaining sections explain how to define the statistics we care about and show that if $s$ and $t$ have identical statistics, then they must have a long LCS.

\smallskip
We now describe the three different high-level strategies we use for finding long common subsequences between $s$ and $t$. The core idea of the matching algorithm is to start by greedily matching ones between $s$ and $t$, and switch opportunistically to matching zeros in certain subintervals where zeros are more common. Although this framework may seem simplistic, note that any common subsequence whatsoever of $s$ and $t$ can be obtained in this fashion, by scanning along the two strings and switching between matching for ones and zeros as convenient. Fairly simple switching strategies based on tracking local bit densities will be sufficient for our proofs.

\medskip\noindent\textbf{Strategy 1: Globally imbalanced strings.} The first strategy is extremely simple: match corresponding ones (or zeros if there are more zeros) in $s$ and $t$. If $s$ and $t$ are significantly imbalanced in the same direction, i.e. they both have more than $(1/2+\delta_0) N$ or fewer than $(1/2-\delta_0)N$ ones, then this strategy immediately finds an LCS of length at least $(1/2+\delta_0)N$. This naive strategy is illustrated in Figure~\ref{fig:imbalanced-case}. It may be helpful to visualize the strategy as sending two runners, one down the length of each string, who must advance simultaneously while holding hands and only step on the ones in their respective strings. 
\begin{figure}[h]
\centering
\vspace{-5mm}
\scalebox{1.5}{
\begin{tikzpicture}
\subseqpic
\drawmatch{0,0}{0,0}
\drawmatch{0,1}{0,2}
\drawmatch{0,3}{0,3}
\drawmatch{0,7}{0,5}
\drawmatch{1,1}{0,6}
\drawmatch{1,2}{1,4}
\drawmatch{1,6}{1,7}
\drawmatch{1,7}{2,0}
\drawmatch{2,0}{2,2}
\drawmatch{2,1}{2,3}
\drawmatch{2,3}{2,4}
\drawmatch{3,0}{2,6}
\drawmatch{3,2}{3,1}
\drawmatch{3,4}{3,2}
\drawmatch{3,6}{3,5}
\drawmatch{3,7}{3,6}
\end{tikzpicture}
}
\vspace{-5mm}
\caption{In the naive strategy, we match ones greedily between $s$ and $t$. In this example, the strings are balanced so the naive strategy finds a common subsequence of length exactly $16$ between $s$ and $t$ of length $32$.}
\label{fig:imbalanced-case}
\end{figure}
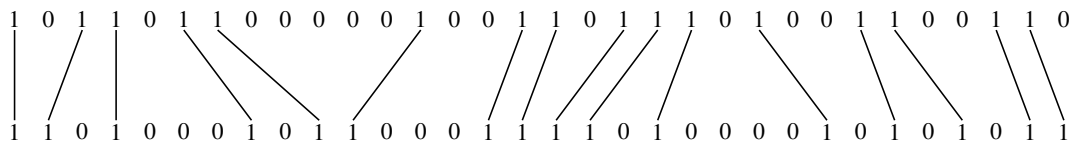

\medskip\noindent\textbf{Strategy 2: Green strategy.}
The other two strategies, which we call the Green strategy and the Blue-Yellow strategy, are both modifications of the naive strategy. We first describe the Green Strategy, which is carried out in Section~\ref{sec:green}, as it is simpler. The naive strategy above can be thought of as a scanning process, where two runners move along the ones in $s$ and $t$ simultaneously, matching bits together to find a common subsequence composed entirely of ones. In the Green strategy, we fix an ``oscillation period'' $\ell \ge 1$ and preprocess $s$ by counting, for every index $i$ of a one-bit in $s$, the number of zeros between the $i$-th one and the $(i+\ell)$-th one. For each such $i$, we plant a marker there, which we call a {\it Green $\ell$-flag}, if there are more than $(1+\eps)(\ell - 1)$ zeros in this interval. Since there are exactly $\ell-1$ ones in this same interval, the Green $\ell$-flags are meant to signal to the runners that they are entering a zero-rich patch within $s$. The same preprocessing is also done in $t$.

In the Green strategy, the two runners proceed in the same way as the naive strategy except that each time the runners reach Green flags simultaneously, they switch to stepping only on zeros for the duration of the flagged regions. As there are more zeros than ones in the flagged regions, they pick up an advantage over the naive strategy for every single pair of Green flags they simultaneously match. The Green strategy is pictured in Figure~\ref{fig:green-case}.

\begin{figure}[h]
\centering
\vspace{-5mm}
\scalebox{1.5}{
\begin{tikzpicture}
\drawpartg{0}{1}{1}{1}{0}{2}{1}{3}
\drawpartgg{0}{1}{0}{1}{0}{2}{0}{2}
\drawpartg{2}{3}{3}{5}{2}{4}{3}{4}
\drawpartgg{2}{3}{2}{3}{2}{4}{2}{4}
\subseqpic
\drawmatch{0,0}{0,0}
\drawmatch{0,1}{0,2}
\drawmatch{0,2}{0,4}
\drawmatch{0,4}{0,7}
\drawmatch{0,5}{1,0}
\drawmatch{0,6}{1,1}
\drawmatch{1,0}{1,2}
\drawmatch{1,2}{1,4}
\drawmatch{1,6}{1,7}
\drawmatch{1,7}{2,0}
\drawmatch{2,0}{2,2}
\drawmatch{2,1}{2,3}
\drawmatch{2,3}{2,4}
\drawmatch{2,4}{2,5}
\drawmatch{2,5}{2,7}
\drawmatch{2,6}{3,0}
\drawmatch{2,7}{3,3}
\drawmatch{3,1}{3,4}
\drawmatch{3,6}{3,5}
\drawmatch{3,7}{3,6}
\end{tikzpicture}
}
\vspace{-5mm}
\caption{In the Green strategy, we find a long common subsequence by matching zeros instead of ones in the zero-rich patches immediately following Green $\ell$-flags. In this example $\ell = 4$, the darker Green one-bits are Green $4$-flags, and the light Green substrings following them are zero-rich patches. Note that to keep the figure clean we did not mark all green $4$-flags in these strings. The Green strategy finds a common subsequence of length $20$ between the $s$ and $t$ pictured above.}
\label{fig:green-case}
\end{figure}
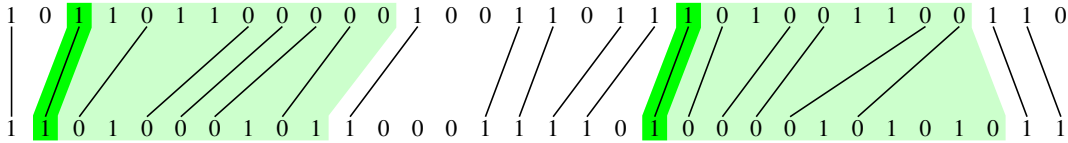

Our analysis of the success of the Green strategy is conditioned on the existence of a single oscillation period $\ell$ for which many of these zero-rich patches exist in both $s$ and $t$. Indeed, suppose there exists $\ell$ for which a constant $g_\ell = \Omega(1)$ fraction of the ones in both $s$ and $t$ are Green $\ell$-flags. Typically, we expect that the two runners hit Green flags simultaneously a constant $g_\ell^2$ fraction of the time (this can be made rigorous by randomly shifting the starting position of one of the runners slightly). Thus, using the Green strategy one can find a common subsequence of length $(1/2 + g_\ell^2 \eps)N$. The Green case finishes the proof of Theorem~\ref{thm:main-intro} if there exists any single oscillation period $\ell$ for which a constant fraction of ones in $s$ and $t$ are Green $\ell$-flags.

Unfortunately, it is not always the case that a string $s$ has a single oscillation period $\ell$ as above. Indeed, if
\[
s_i \defeq (1^{2^i} 0^{2^i})^{2^{k-i-1}},
\]
then each $s_i$ is a string of length $2^k$ which oscillates with period $2^i$. It is not hard to check that the concatenation $s\defeq s_0 s_1 \cdots s_{k-1}$ is a string of length $N=k\cdot 2^k$, such that there are at most $O(2^k)$ Green $\ell$-flags in $s$ for any given choice of $\ell$. Thus, $g_\ell = O((\log N)^{-1}) = o(1)$ for every single $\ell$, so the Green strategy is insufficient for this type of string (and there are many strings with this property). We remark that this is essentially the worst case: it is not difficult to show by the machinery in Section~\ref{sec:structure} that for any locally-balanced string $s$, $\sum g_\ell \ge \Omega(N)$ where the sum is over scales $\ell$ equal to a power of two. Thus, by the pigeonhole principle at least $\Omega(N/\log N)$ of these flags appear at the same scale $\ell$. One can always find two strings $s$ and $t$ in $C$ with $g_\ell = \Theta((\log N)^{-1})$ with the same $\ell$, proving $\LCS(s,t) \ge (1/2 + \Omega((\log N)^{-2})) N$ using the Green strategy alone. Already, this argument saves several factors of $\log N$ in the surplus term over the argument of Bukh and Ma \cite{bukh_ma}.

\medskip\noindent\textbf{Strategy 3: Blue-Yellow strategy.}
We give a third and final strategy which handles the cases in which $g_\ell = o(1)$ for all oscillation periods $\ell$, which we call the Blue-Yellow strategy. This strategy, handled in Section~\ref{sec:blue}, is the most involved of the three and we do not explain all of the technical complications here. However, the general picture is similar to the Green strategy: we send two runners along $s$ and $t$ matching ones, and find opportune moments to switch to matching zeros to gain an advantage.

\begin{figure}[h]
\centering
\vspace{-5mm}
\scalebox{1.5}{
\begin{tikzpicture}
\drawpartby{0}{7}{2}{0}{0}{5}{1}{3}
\drawpartbyby{0}{7}{0}{7}{0}{5}{0}{5}
\drawpartyb{2}{3}{3}{1}{1}{7}{3}{0}
\drawpartybyb{2}{3}{2}{3}{1}{7}{1}{7}
\subseqpic
\drawmatch{0,0}{0,0}
\drawmatch{0,1}{0,2}
\drawmatch{0,7}{0,5}
\drawmatch{0,3}{0,3}
\drawmatch{1,0}{0,7}
\drawmatch{1,3}{1,0}
\drawmatch{1,4}{1,1}
\drawmatch{1,5}{1,2}

\drawmatch{2,1}{1,4}
\drawmatch{2,3}{1,7}
\drawmatch{2,4}{2,1}
\drawmatch{2,5}{2,5}
\drawmatch{2,6}{2,7}
\drawmatch{2,7}{2,8}

\drawmatch{3,2}{3,1}
\drawmatch{3,4}{3,2}
\drawmatch{3,6}{3,5}
\drawmatch{3,7}{3,6}
\end{tikzpicture}
}
\vspace{-5mm}
\caption{In the Blue case, we find a long common subsequence by matching zeros from extremely zero-rich patches (signalled by Blue flags) to longer relatively balanced patches (signalled by Yellow flags) and vice-versa. In the figure, a Blue $2$-flag from each of the strings is matched with a Yellow $6$-flag from the other. The Blue-Yellow strategy finds an LCS of length $18$ between the $s$ and $t$ pictured above.}
\label{fig:blueyellow-case}
\end{figure}
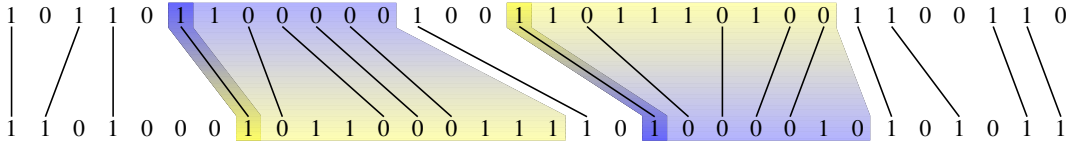

In the Blue-Yellow strategy, we also mark certain one-bits in $s$ and $t$ by flags, but we use flags of two different colors Blue and Yellow. A Blue $\ell$-flag is a relatively rare occurrence: it signals that there is an extremely zero-rich interval afterwards containing $\ell-1$ ones and more than $\eps^{-1}(\ell-1)$ zeros. A Yellow $\ell$-flag, on the other hand, is very common: it signals there is an interval afterwards containing $\ell-1$ ones and more than $0.9 (\ell-1)$ zeros. Also, since there is no single oscillation period that captures the behavior of $s$ and $t$ (or else we would apply the Green strategy), we must pay attention to flags at many different scales $\ell$ at the same time. The rough idea is then that the runners will switch to matching zeros when one of them reaches a Blue flag, and the other one simultaneously hits a Yellow flag of a similar scale, see Figure~\ref{fig:blueyellow-case}.  

In the above diagram, the top runner reaches a Blue flag first, and the bottom runner happens upon a Yellow flag at the same time. This signals both of them to switch to matching zeros, which allows the top runner to gain a great advantage (since the patch past a Blue flag is so zero-rich). The bottom runner may lose out slightly in efficiency because Yellow intervals can have slightly fewer zeros than ones, but on net we see that more bits are used from the two patches together than would have been otherwise. In the long run, we expect Blue flags to appear approximately equally frequently in $s$ and $t$, so the advantages and losses balance out to a net gain on both sides.

 Now we explain roughly how the Blue-Yellow strategy circumvents the obstacle that the Green strategy ran into. The Green strategy by itself cannot prove Theorem~\ref{thm:main-intro}, since there exists strings $s$ such as $s=s_0 \ldots s_{k-1}$ where $s_i$ oscillates with period $2^i$, so that $s$ has only $O(N/\log N)$ Green $\ell$-flags at any single scale $\ell$. Thus, it is insufficient to consider flags at only a single scale. To take a concrete example, supposing $t = s_{\sigma(0)}\ldots s_{\sigma(k-1)}$ for a typical permutation $\sigma$, the Green strategy fails to find an common subsequence of length more than $(1/2 + \Omega((\log N)^{-2}))N$ between $s$ and $t$. In this example, instead of focusing on a single $\ell$, we instead consider all one-bits in $s$ and $t$ that are the Blue $\ell$-flags for some $\ell \ge N^{1-\eps}$. This threshold $N^{1-\eps}$ is chosen so that there are $\Theta(\eps^2 N)$ such flags in each of $s$ and $t$. Because these are the flags at the largest $\eps$-fraction of scales, we see that for every single $\ell \ge N^{1-\eps}$, most $\ell$-intervals in $s$ and $t$ have density close to $1/2$, so most $\ell$-flags in $s$ and $t$ are Yellow. As a result, we expect that as the two runners scan through $s$ and $t$, most Blue $\ell$-flags in one will be matched with a Yellow flag in the other, resulting in favorable situations as in Figure~\ref{fig:blueyellow-case}. The Blue-Yellow strategy succeeds by accumulating all these advantages across the $\Omega(\eps N)$ Blue flags in each of $s$ and $t$, to find an LCS of length $(1/2 + \Omega(1))N$.

\medskip\noindent\textbf{String regularity.}
In order to guarantee that the two runners remain relatively synchronized as the Blue-Yellow strategy proceeds, we need Blue flags to be somewhat consistently distributed within $s$ and within $t$. 
We get this desired property by proving a string regularity lemma similar to that of~\cite{APP13} (see also \cite{BZ16,ACKPS19, HKP21}).
Our regularity lemma, proved in Section~\ref{sec:entropy} differs from previous versions for words in that we use an \emph{entropy increment} argument, rather than the usual density increment argument.

\medskip\noindent\textbf{The structure lemma.}
The remainder of the paper is designed to set up the strings $s$ and $t$ so that one of the above three strategies can succeed in finding a long common subsequence.
To do this, we prove a structure lemma in Section~\ref{sec:structure} about strings, which says each string falls in one of three cases (1) Globally imbalanced, (2) Green at some oscillation period $\ell$, or (3) Blue-Yellow at some oscillation period $\ell$ (for the formal definition of the Blue-Yellow case, see Definition~\cite{def:type}).
These types are defined such that, if two strings $s$ and $t$ have the same type and same oscillation period (if applicable), one can find a long common subsequence of $s$ and $t$ using the corresponding strategy.

With this structure lemma, we simply need to put all the pieces together (Section~\ref{sec:wrapup}).
We partition each codeword $s$ into $\poly\log N$ substrings $s_1,\dots,s_{\poly\log N}$ and apply the structure lemma to each substring $s_i$.
By pigeonhole, there exist two strings $s$ and $t$ such that, for each $i=1,\dots,\poly\log N$, the substrings $s_i$ and $t_i$ are the same ``type,'' meaning that they are in the same case of the structure lemma and have the same oscillation period (if applicable).
Then in each pair of substrings $s_i$ and $t_i$ one can find a long common subsequence using one of the three strategies, giving an overall large LCS.

It is important for two technical reasons to split into substrings $s_i$ and $t_i$, rather than applying the structure lemma directly to the entire strings $s$ and $t$. 
First, we need to randomly shift the starting position of one of the ``runners'' to get enough synchronized flags, incurring a loss of order up to $|s_1|$, so it is necessary that this loss is $o(N)$.
Second, our regularity lemma guarantees Blue flags to be well-distributed at most, but not all, scales.
Thus, we may not obtain the desired regularity until we consider substring lengths down to around $O(N/\poly \log N)$.

\medskip\noindent\textbf{Organization.} Section~\ref{sec:prelim} collects the common notations, definitions, and preliminary lemmas we need for the rest of the paper. In Section~\ref{sec:structure}, we prove a structure lemma which divides strings into three types, one suitable to each of the three strategies above. In Section~\ref{sec:entropy}, we perform an additional technical argument to prove a ``regularity-type'' property of strings necessary for the runners to remain synchronized in the Blue-Yellow case (so that one does not race too far ahead of the other).  These two sections together set the stage for Sections~\ref{sec:green},~\ref{sec:blue}, and~\ref{sec:wrapup} to handle the Green case, the Blue-Yellow case, and complete the proof, respectively. Finally, we collect some of the tantalizing open questions that remain in this area in Section~\ref{sec:concluding}.

\section{Preliminaries and Notation}\label{sec:prelim}

\smallskip\noindent \textbf{Constants.} Throughout, fix $\varepsilon=10^{-6}$ and $\gamma=10^{-15} = 0.001\varepsilon^2$.

\medskip\noindent\textbf{Subsequences and substrings.} A {\it subsequence} in a string $s$ is any string obtained from $s$ by deleting zero or more symbols. In contrast, a substring is a subsequence consisting of consecutive symbols of $s$ (substrings are also sometimes referred to as {\it subwords} elsewhere, but we do not use this terminology). Thus, if $s=10001$, then $101$ is a subsequence of $s$ but not a substring. 

\medskip\noindent\textbf{Intervals and strings.}
Throughout, for real numbers $x$ and $y$ we define an \emph{interval} $I=[x,y]$ to be the set of integers $a$ such that $x\le a\le y$ (rather than the set of real numbers $a$).
We let $[n]$ denote the interval $[1,n]$.
We similarly define intervals $(x,y]$ and $[x,y)$ and $(x,y)$ as subsets of the integers. 
The \emph{size} of an interval is the number of integers in the interval.
For $\alpha\in(0,1)$ and real number $x$, let $(1\pm \alpha)x$ denote the interval $[(1-\alpha)x,(1+\alpha)x]$.

Throughout the remainder of the paper, all strings are binary.
Let $\{0,1\}^*$ be the set of all binary strings of any nonnegative length.
For integers $m\ge 0$ and $i\ge 1$, we let $I_{m,i}\defeq[(i-1)\cdot2^{m}+1,i\cdot2^{m}]$.
We call such an $I_{m,i}$, where both the size is a power of two and the endpoints are aligned with the same power of two, a \emph{dyadic interval}.
For strings $s_1$ and $s_2$, let $s_1s_2$ denote the string concatenation of $s_1$ with $s_2$.

For a string $s$ with $L$ ones and an interval $I=[x,y]$, let $s_{I}$ denote the contiguous substring of $s$ between the $\ceil{x}$-th one of $s$ (or the beginning of the string if $x\le 0$) and the $\floor{y+1}$-st one of $s$ (or until the end of the string if $y\ge L$), including the first but excluding the second. For example, if $w=1001011$, we have $w_{[1,2]} = 10010$.
For $m\ge0$ and $i\ge 1$ write $s_{m,i}$ as shorthand for $s_{I_{m,i}}$.
Informally, we refer to $s_{m,1},s_{m,2},s_{m,3},\dots$ as the \emph{substrings of $s$ at scale $m$}.
We note that leading zeros of a string are not included in any dyadic substring $s_{m,i}$, but this is negligible as we typically work with strings that start with a one.

We write $z(s)$ for the number of zeros in a string $s$. 

\medskip\noindent\textbf{Reversing strings that begin with a one.}
In the Blue-Yellow strategy on two strings $s$ and $t$, we apply the Blue-Yellow matchings in pairs: (i) matching Blue flags in $s$ with Yellow flags in $t$, and (ii) matching Blue flags in $t$ with Yellow flags in $s$.
Arguments (i) and (ii) can be handled with a single argument when (ii) is viewed as applying the lemma on the reversals of $s$ and $t$.
However, since throughout we index our substrings $s_I$ by the one-bits rather than all bits, it is helpful to slightly modify the definition of string reversal as follows.
Given a string $s$ starting with a one, let $\rev(s)$ denote the string obtained by reversing the order of all the bits in $s$ after the first bit. Thus, $\rev(s)$ is only defined for strings starting with a one.
It is easy to check the following properties of $\rev$.

\begin{lemma}
  \label{lem:rev}
  Let $w$ be a string that begins with a one and has 
  $L$ ones in total.
  \begin{enumerate}
  \itemsep=0ex
  \item The strings $w$ and $\rev(w)$ have the same length and number of ones.
  \item For an interval $I=[x,y]\subset[L]$, we have $\rev(w_I) = \rev(w)_{[L+1-y,L+1-x]}$.
  \item When $w$ has $L=2^n$ ones we have $\rev(w_{m,i}) = \rev(w)_{m,2^{n-m}+1-i}$.
  \end{enumerate}
\end{lemma}
\begin{proof}
  The first item is obvious.
  For the second item, it suffices to consider when $x$ and $y$ are integers: indeed, for real numbers $x$ and $y$, we have $[x,y]=[\ceil{x},\floor{y}]$ and $[L+1-y,L+1-x] = [L+1-\floor{y},L+1-\ceil{x}]$, so we may replace $x$ and $y$ with $\ceil{x}$ and $\floor{y}$.
  The zeros between the $i$-th and $(i+1)$-st one of $w$ map to the zeros between the $(L+1-i)$-th and $(L+2-i)$-th one of $\rev(s)$.
  Let $z_i$ denote the number of zeros between the $i$-th and $(i+1)$-st one of $w$.
  Then
  \begin{equation*}
  \rev(w_{[x,y]}) = \rev(10^{z_x}10^{z_{x+1}}1\cdots 10^{z_y}) = 10^{z_y}10^{z_{y-1}}1\cdots 10^{z_x}=\rev(w)_{[L+1-y,L+1-x]}.
  \end{equation*}
  The third item follows from the second:
  \begin{align}
    \rev(w_{m,i})
    = \rev(w_{[(i-1)\cdot 2^m+1, i\cdot 2^m]})
    &= \rev(w)_{[2^n+1 - i\cdot 2^m , 2^n - (i-1)\cdot 2^m]} \nonumber\\
    &= \rev(w)_{[2^m (2^{n-m}-i)+1, 2^m\cdot(2^{n-m}-i+1)]}
    =\rev(w)_{m,2^{n-m}+1-i}.\qedhere
  \label{}
  \end{align}
\end{proof}
\begin{example}
  \label{ex:1}
  For $w=1001011$, we have
  \begin{equation*}
  \rev(w_{[1,2]}) = \rev(10010) = 10100 = (1110100)_{[3,4]} = \rev(w)_{[3,4]}.
  \end{equation*}
\end{example}
\begin{lemma}($\rev$ preserves LCS)
  For strings $s$ and $t$ starting with a one, $\LCS(s,t) = \LCS(\rev(s),\rev(t))$.
\label{lem:rev-lcs}
\end{lemma}
\begin{proof}
  The LCS always matches the first bits if they are equal, and reversing strings preserves the LCS.
\end{proof}

\paragraph{Flags.} We now define flags, a key notion that measures the oscillation frequencies within string.
\begin{definition}[Flags]
For a positive integer $\ell$ and a string $w$, define an index $i\in\ZZ$ to be {\it an $\ell$-flag of rate $r$ in $w$} if $(\ell-1)^{-1}z(w_{[i,i+\ell)})=r$.
Here $0\le r \le +\infty$ is a real number that can take on the value $+\infty$ and we define $0^{-1}\cdot0=0$ and $0^{-1}\cdot m=+\infty$ for any positive integer $m$. 
The rate of an $\ell$-flag $i$ in $w$ is the ratio of zeros to ones (strictly) between the $i$-th one and the $(i+\ell)$-th one of $w$, and we would like to find $\ell$-flags with high rate in order to execute the zero-matching strategies described in the previous section. 
We say $\ell$-flag $i$ of rate $r$ is 
\[
\begin{cases}
\textnormal{Blue} & \textnormal{if } r > \eps^{-1}, \\
\textnormal{Green} & \textnormal{if } r > 1+2\eps, \\
\textnormal{Yellow} & \textnormal{if } r > 0.9, \\
\textnormal{Red} & \textnormal{if } r \le 0.9. \\
\end{cases}
\]
Note that Blue flags are Green flags and Green flags are Yellow flags. 
For a string $w$ with $L$ ones, for each $i\in[L]$, define $b_w(i)$ to be the value $\ell\in [L]$ such that $\ell$ is the largest power of two for which $i$ is a Blue $\ell$-flag in $w$, with $b_w(i) = 0$ if no such $\ell$ exists.
We say $i$ is a Blue $\ell^+$-flag in $w$ if $b_w(i)\ge \ell$.
\end{definition}

\medskip\noindent\textbf{Imbalanced strings.}
For $\delta \in (0, 1/2)$, we say a string $w$ with $L$ ones is \emph{$\delta$-imbalanced} if its number of zeros $z(w)$ is not in $(1\pm\delta)L$. For convenience, we simply say that $w$ is {\it imbalanced} if it is $\eps$-imbalanced with $\eps = 10^{-6}$ defined at the beginning of this section. We will reason about imbalanced substrings $w_I$ of $w$ at various scales, and exploit the existence of these imbalanced substrings.
\begin{lemma}
  Let $w$ be a string with $L$ ones, and let $\ell\ge 2\varepsilon^{-1}$. 
  Suppose that $i$ is a Blue or Green $\ell$-flag in $w$, or $i\le L-\ell+1$ is a Red $\ell$-flag.
  On the interval $I=[i,\min(i+\ell-1,L)]$, the substring $w_I$ is imbalanced. 
  \label{lem:blue-0}
\end{lemma}
\begin{proof}
  If $i$ is a Green $\ell$-flag, substring $w_{I}$ has at least $(1+2\varepsilon)(\ell-1) > (1+\varepsilon)\ell \ge (1+\varepsilon)|I|$ zeros.
  Since Blue flags are Green flags, $w_I$ is also imbalanced if $i$ is a Blue $\ell$-flag.
  If $i$ is a Red $\ell$-flag with $i\le L- \ell + 1$, then $|I| = \ell$ and substring $w_I$ has at most $0.9(\ell-1) < (1-\varepsilon)|I|$ zeros. In all three cases, $w_I$ is imbalanced, as desired.
\end{proof}

It is also easy to see that if two imbalanced strings have the same length $n$ and the same number of ones, then their LCS is a constant fraction larger than $n/2$. 
\begin{lemma}
  Let $\delta\in(0,1/2)$.
  Let $s$ and $t$ be $\delta$-imbalanced strings of the same length with the same number $L$ of ones in each. Then $\LCS(s,t)\ge (1/2+\delta/5)|s|$.
\label{lem:trivial-0}
\end{lemma}
\begin{proof}
  If the number of zeros is at most $(1-\delta)L$, then the all-ones string is a common subsequence of length $L\ge \frac{1}{2-\delta}|s| \ge (1/2+\delta/5)|s|$.
  If the number of zeros is at least $(1+\delta)L$, then the all-zeros string is a common subsequence of length at least $\frac{1+\delta}{2+\delta}|s| \ge (1/2+\delta/5)|s|$.
\end{proof}

\paragraph{Prefixes and suffixes.}
For a string $w$ with $L$ ones, and $\Delta\in[-L,L]$, let $\Drop_\Delta(w) \defeq w_{[\max(1,1-\Delta),\min(L,L-\Delta)]}$. Thus, $\Drop_\Delta(w)$ is a prefix of $w$ if $\Delta \ge 0$ and a suffix otherwise. The following lemma (see Figure~\ref{fig:shift}) shows that finding long common subsequences across prefixes and suffixes of many subintervals of $s$ and $t$ implies that $\LCS(s,t)$ is large overall.

\begin{figure}[h]
\centering
\pgfdeclareverticalshading{rainbow}{100bp}
{color(0bp)=(yellow); color(30bp)=(yellow); color(50bp)=(green); color(70bp)=(blue); color(100bp)=(blue)}
\newcommand\addZShade[1]{
  \fill[violet!30] (#1-0.6,0)--(#1-1,1) -- (#1-0.4,1) -- (#1,0);
  \node at (#1-0.5,0.5) {\tiny $+\delta \cdot 2^m$};
\node[outer sep=-5, label={[anchor=south,gray]\tiny$\Drop_\Delta(s_{m,#1})$}] at (#1-0.7,1.05) {};
\node[outer sep=-1,                           label={[anchor=north,gray]\tiny$\Drop_{-\Delta}(t_{m,#1})$}] at (#1-0.3,-0.05) {};
\draw [decorate,decoration={brace,amplitude=3pt}] (#1-1,1) -- (#1-0.4,1);
\draw [decorate,decoration={brace,amplitude=3pt,mirror}] (#1-0.6,0) -- (#1,0);
}
\newcommand\addNormalShade[2]{
  \fill[gray!30] (#2-0.6,0)--(#2-1,1) -- (#1-0.4,1) -- (#1,0);
}
\vspace{-5mm}
\scalebox{1}{
\begin{tikzpicture}
\node[label={[anchor=east]$t$}] at (0,0) {};
\node[label={[anchor=east]$s$}] at (0,1) {};
\draw (0,0) -- (16,0);
\draw (0,1) -- (16,1);
\foreach \x in {0,...,16} {
  \draw(\x,0.1)--(\x,-0.1);
  \draw(\x,1.1)--(\x,0.9);
}

\node[outer sep=-5, label={[anchor=south,gray]$s_{m,3}$}] at (2.5,2) {};
\node[outer sep=-1,                           label={[anchor=north,gray]$t_{m,3}$}] at (2.5,-1) {};
\node[outer sep=-5, label={[anchor=south,gray]$s_{m,2^{n-m}}$}] at (15.5,2) {};
\node[outer sep=-1,                           label={[anchor=north,gray]$t_{m,2^{n-m}}$}] at (15.5,-1) {};
\draw[dashed,gray] (2,2) rectangle (3,-1);
\draw[dashed,gray] (15,2) rectangle (16,-1);
\addNormalShade{0.4}{3}
\addNormalShade{3}{6}
\addNormalShade{6}{8}
\addNormalShade{8}{14}
\addNormalShade{14}{16.6}
\addZShade{3}
\addZShade{6}
\addZShade{8}
\addZShade{14}

\toplabelshift{0}{1}
\toplabelshift{1}{2^m}
\toplabelshift{15.6}{2^n-\Delta}
\bottomlabelshift{0}{1}
\bottomlabelshift{0.4}{\Delta}
\bottomlabelshift{1}{2^m}
\bottomlabelshift{16}{2^n}

\end{tikzpicture}
}
\vspace{-5mm}
\caption{If, for some $\Delta$, we find a large LCS between many (prefixes and suffixes of) dyadic substrings of $s$ and $t$, then $\LCS(s,t)$ is large overall.
We note that Lemma~\ref{lem:shift} works as long as the subintervals (in purple) have LCS beating the trivial matching by $\delta\cdot 2^m$ \emph{on average}, even though the figure depicts \emph{each} subinterval having LCS advantage $\delta\cdot 2^m$.  In the diagram, the set $Z$ from Lemma~\ref{lem:shift} is $\{3,6,8,14\}$. \protect\footnotemark}
\label{fig:shift}
\end{figure}
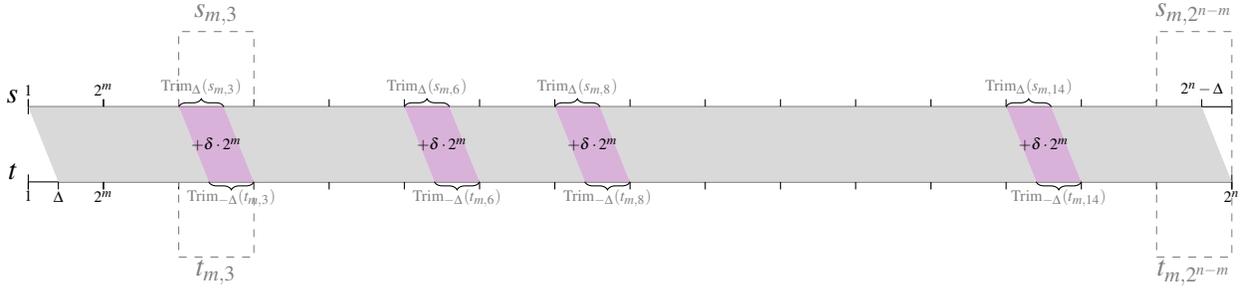
\footnotetext{Technically, the figure is invalid because we need $m\le n-10-\log\delta^{-1}$ and here $m=n-4$, but we ignore this issue for illustration.}
\begin{lemma}[Prefix/Suffix LCS]
\label{lem:shift}
  Let $\delta>0$, let $m$ and $n$ be integers with $0\le m\le n-10-\log\delta^{-1}$ and $L=2^n$, and let $Z\subset[2^{n-m}]$ satisfy $|Z|\ge 2^{n-m}/10$.
  Suppose that $s$ and $t$ are strings with $L$ ones each, and there exists $\Delta\in [-2^m,2^m]$ such that
  \begin{equation*}
    \sum_{i\in Z}^{} \LCS(\Drop_\Delta(s_{m,i}),\Drop_{-\Delta}(t_{m,i}))\ge |Z|\cdot (2^m - |\Delta| + \delta \cdot 2^m).
  \end{equation*} 
  Then we have
  \begin{equation*}
    \LCS(s,t)\ge \left(1 + \frac{\delta}{20}\right)L.
  \end{equation*}
\end{lemma}
\begin{proof}
  Finding a common sequence between $s$ and $t$ is equivalent to exhibiting a matching between the bits of $s$ and the bits of $t$ such that only equal bits are matched and such that the matching is ``non-crossing,'' meaning that earlier bits of $s$ are matched with earlier bits of $t$.

  Consider the matching where we match the $i$-th one in $s$ with the $(i+\Delta)$-th one in $t$.
  This matches $L-|\Delta|$ ones.
  In this matching, for each $i\in Z$, the $2^m-|\Delta|$ ones of $\Drop_\Delta(s_{m,i})$ are exactly matched to the $2^m-|\Delta|$ ones of $\Drop_{-\Delta}(t_{m,i})$.
  For each $i\in Z$, replace the matching between the ones of $\Drop_\Delta(s_{m,i})$ and $\Drop_{-\Delta}(t_{m,i})$ with a matching for the LCS of $\Drop_\Delta(s_{m,i})$ and $\Drop_{-\Delta}(t_{m,i})$. 
 
  All of these replacements can be done simultaneously and independently while keeping the matching non-crossing.
  Each of the $|Z|$ replacement operations deletes $2^m-|\Delta|$ pairs, and in total the replacements add $|Z|(2^m-|\Delta|+\delta \cdot 2^m)$ pairs.
  Thus in total the replacements increase the number of matched pairs by at least $|Z|\cdot \delta \cdot 2^m$.
  Thus, the total length of this common subsequence is at least (recall $L=2^n$)
  \begin{align*}
    L-|\Delta| + |Z|\cdot (\delta 2^m) 
    \ge L - 2^{m} + \frac{2^{n-m}}{10}\cdot \delta\cdot 2^m
    \ge L - \frac{\delta}{2^{10}}L + \frac{\delta}{10} L
    \ge \left(1 + \frac{\delta}{20}\right)L.
  \end{align*}
  In the second inequality, we used that $m\le n-10-\log\delta^{-1}$.
  Thus $\LCS(s,t)\ge \left(1 + \frac{\delta}{20}\right)L$, as desired. 
\end{proof}

\section{The structure lemma and definition of types}\label{sec:structure}

\subsection{The structure lemma}

Throughout this section, we reason about a single string $w$ with $L$ ones, and assume $w$ starts with a one. 
Recall that $\varepsilon=10^{-6}$.
Our main structure lemma is as follows.
 
\begin{lemma}[Structure Lemma]
\label{lem:classify}

If $w\in\{0,1\}^{*}$ is a string that starts with one and has exactly $L=2^n$ ones, and $n$ is sufficiently large (in terms of $\eps$), then at least one of the following conditions hold.
\begin{enumerate}
\item \label{enu:skew}There exists an interval $I\subseteq[L]$ of size $|I| \ge \eps^2 L$ such that $w_I$ is imbalanced.
\item \label{enu:Green}There exists $1 \le \ell \le L$ such that the number of Green $\ell$-flags in $w$ is
at least $\eps^{2}L$.
\item \label{enu:Blue-Red}There exists $1\le m\le n$ such that the number of Blue $(2^m)^+$-flags in $w$ is
at least $\eps^{2}L$, and for every $\ell \ge 2^m$ the number of Red $\ell$-flags in $w$ is at most $600\eps L$.
\end{enumerate}
\end{lemma}
The three cases of the Structure Lemma correspond exactly to the three matching strategies outlined in the Overview (Section~\ref{sec:overview}).
Case \ref{enu:skew} is when $w$ is imbalanced at a macroscopic scale, i.e., a linear-length subword with density far from $\frac{1}{2}$. 
Case \ref{enu:Green}, the ``Green case'', is when $w$ ``fluctuates on a single scale'' and can be treated by studying that scale only, i.e., using the Green strategy described in the Overview (see Section~\ref{sec:green}).
Case~\ref{enu:Blue-Red} is when $w$ is ``sporadic'' and must be analyzed at many scales simultaneously, which is done with the Blue-Yellow strategy in the Overview (see Section~\ref{sec:blue}). 
Most strings $w$ in fact fall into Case \ref{enu:Green} with $\ell=1$, and Case \ref{enu:Blue-Red} is only necessary for highly structured strings designed to oscillate at multiple frequencies.
Before we prove Lemma~\ref{lem:classify}, we need the following technical lemma, which gives part of Lemma~\ref{lem:classify} under the additional assumption that many of the zeros in $w$ are concentrated at Blue flags.

\begin{lemma}
\label{lem:blue-plus}
Suppose $\alpha > 0$, $n$ is sufficiently large, $L=2^{n}$, $1\le \ell \le L$, and $w\in\{0,1\}^{*}$ is a string with $L$ ones. If
\begin{align}
\frac{1}{\ell}\sum_{i\in B_{\ell}}z(w_{[i,i+\ell)})\ge\alpha L,
  \label{eq:blue-plus-0}
\end{align}
where $B_{\ell}$ is the set of all Blue $\ell$-flags of $w$, then at least one of the following conditions hold. 
\begin{enumerate}
\item \label{enu:bp-skew}There exists an interval $I\subseteq[L]$ of size $|I| \ge \eps^2 L$ such that $w_I$ is imbalanced.
\item \label{enu:bp-Green} The number of Blue $\ell$-flags in $w$ is
at least $\eps^{2}L$.
\item \label{enu:bp-Blue-Red} The number of Blue $\ell^{+}$-flags in $w$ is at least $(\alpha-22\eps)\eps L/16$.
\end{enumerate}
\end{lemma}

\begin{proof} We assume Conditions~\ref{enu:bp-skew} and~\ref{enu:bp-Green} do not hold and prove Condition~\ref{enu:bp-Blue-Red}.
We first prove a slightly stronger lower bound (see \eqref{eq:t-lower-power-two}) on the number of Blue $\ell^+$ flags than claimed in Condition~\ref{enu:bp-Blue-Red} in the special case that $\ell$ is a power of $2$, and will then handle the case of general $\ell$.

Since $|B_{\ell}|<\eps^{2}L$, if we define $A_{\ell}\defeq\{i\in B_{\ell} \mid z(w_{[i,i+\ell)})>2\eps^{-1}\ell\}$
we find that
\[
\frac{1}{\ell}\sum_{i\in B_{\ell}\backslash A_{\ell}}z(w_{[i,i+\ell)})
\le
\frac{1}{\ell}\cdot (2\varepsilon^{-1}\ell)\cdot |B_\ell\setminus A_\ell|
\le
2\eps L,
\quad 
\text{so}
\quad
\frac{1}{\ell}\sum_{i\in A_{\ell}}z(w_{[i,i+\ell)})\ge(\alpha-2\eps)L \ .
\]
By the pigeonhole principle, we can pick a residue class $r\in\{1,\ldots,\ell\}$
such that the set $S \subseteq [L/\ell]$ defined by $S\defeq\{j \mid (j-1)\ell+r\in A_{\ell}\}$ satisfies 
\[
\sum_{j\in S}z(w_{[(j-1)\ell+r,j\ell+r)})\ge(\alpha-2\eps)L.
\]
Write $z_{j}\defeq z(w_{[(j-1)\ell+r,j\ell+r)})$ for the number of zeros in the $j$-th such interval. We have that for
each $j\in S$, $z_{j}>2\eps^{-1}\ell$, and $z_{S}\ge(\alpha-2\eps)L$
(using the summation notation $z_{U}\defeq\sum_{j \in U}z_{j}$ for any set $U$). 

Let $L'=L/\ell$. Since $\ell$ and $L$ are powers of 2, $L'$ is as well. We consider the family of all dyadic intervals $I_{m,i} \defeq [(i-1) \cdot 2^m + 1, i \cdot 2^m]$ where $0\le m \le \log_2 (L')$ and $1\le i \le L'\cdot 2^{-m}$. We define intervals $I_{m,i}$ for convenience to reason about intervals $J_{m,i}\defeq[(i-1)\cdot2^{m}\ell+r,i\cdot2^{m}\ell+r)$ in the string $w$: indeed, $z_{I_{m,i}} = \sum_{j\in I_{m,i}}^{} z_j$ counts the number of zeros in the substring $w_{[(i-1)\cdot 2^m \ell + r, i\cdot 2^m \ell  + r)} = w_{J_{m,i}}$.

Let $\mathcal{I}$ be the set of all such intervals $I_{m,i}$ maximal under the property that $z_{I_{m,i}}>2\varepsilon^{-1}\ell\cdot |I_{m,i}|$, so that any other dyadic interval satisfying this property lies inside a member of $\mathcal{I}$.
By maximality of the $I_{m,i}$, the intervals of $\mathcal{I}$ are pairwise disjoint.
Furthermore, for $j\in S$, we have $z_{I_{0,j}}=z_j>2\varepsilon^{-1}\ell$, so each $j\in S$ is in some interval of $\mathcal{I}$.

We claim that these intervals satisfy $z_{I_{m,i}}\le4\eps^{-1}\ell\cdot|I_{m,i}|$. Suppose otherwise. Either $I_{m,i} = [L']$ or there is a dyadic interval $I_{m+1,\lceil i/2\rceil}$ in $[L']$ containing $I_{m,i}$ with twice the size of $I_{m,i}$. In the former case, $z(w_{[L]}) = z_{[L']} > 4\varepsilon^{-1}\ell L' > L+\eps L$, which would imply condition~\ref{enu:bp-skew}. In the latter case, $I_{m+1,\lceil i/2\rceil}$ is an interval containing $I_{m,i}$ satisfying $z_{I_{m+1,\lceil i/2\rceil}}>2\eps^{-1}\ell\cdot|I_{m+1,\lceil i/2\rceil}|$, contradicting the maximality of $I_{m,i}$. This proves the claim.

Since every $j\in S$ lies in some element of $\I$, we have
\[
4\eps^{-1}\ell\cdot\sum_{I_{m,i}\in\I}|I_{m,i}|\ge\sum_{I_{m,i}\in\I}z_{_{I_{m,i}}}\ge z_{S}\ge(\alpha-2\eps)L,
\]
and so 
\begin{equation}
\label{eq:length-of-maximal}
    \sum_{I_{m,i}\in\I}|I_{m,i}|\ge\frac{1}{4}(\alpha-2\eps)\eps L/\ell \ . 
\end{equation}

\begin{figure}
  \centering
  \scalebox{1.5}{
\begin{tikzpicture}
  \node at (4, 1) {$w_{J_{m,i-1}}$};
  \node[align=center] at (4, 0.5) {\tiny $2^m\ell$ ones};
  \node at (7.5, 1) {$w_{J_{m,i}}$};
  \node[align=center] at (7.5, 0.65) {\tiny $2^m\ell$ ones};
  \node[align=center] at (7.5, 0.35) {\tiny $z_{I_{m,i}}>2\varepsilon^{-1}\cdot 2^m\ell$ zeros};
  \node[outer sep=0, inner sep=1, fill=blue!30] (1a) at (3.1,0) {\tiny 1};
  \node[outer sep=0, inner sep=1, fill=blue!30] (1b) at (3.5,0) {\tiny 1};
  \node[outer sep=0, inner sep=1, fill=blue!30] (1c) at (3.9,0) {\tiny 1};
  \node[outer sep=0, inner sep=1, fill=blue!30] (1d) at (4.9,0) {\tiny 1};
  \node[outer sep=0, inner sep=0] (b) at (4,-1) {\small Blue $2^{m+1}\ell$-flags};
  \node[outer sep=0, inner sep=1] (1s) at (7.5,0) {\tiny 10000000000000000000000000010100000000001000000};
  \draw[->] (b) -- (1a);
  \draw[->] (b) -- (1b);
  \draw[->] (b) -- (1c);
  \draw[->] (b) -- (1d);
  \draw (3,-0.1) rectangle (5,0.1);
  \draw (5,-0.1) rectangle (10,0.1);
\end{tikzpicture}
}
  \caption{Our method for finding Blue $\ell^+$-flags: the ones in an interval $J_{m,i-1}$ preceeding a zero-rich interval $J_{m,i}$ are all Blue $\ell^+$-flags. We find lots of Blue $\ell^+$-flags by finding pairwise-disjoint zero-rich intervals $J_{m,i}$ with large total size, and then showing their predecessor intervals $J_{m,i-1}$ also have a large union, hence giving lots of Blue $\ell^+$-flags.}
  \label{fig:struct}
\end{figure}

Thus the dyadic intervals in $\I$ have an abundance of zeros in total. We need to convert this into an abundance of $\ell^{+}$-flags. 
The key observation (see Figure~\ref{fig:struct}) is that if $I_{m,i}\in\I$, then any $j\in J_{m,i-1}$ is a Blue $\ell^{+}$-flag. 
Indeed, for such a $j$ we have $[j,j+2^{m+1}\ell)\supseteq J_{m,i}$, so (recalling $z_{I_{m,i}}=z(w_{J_{m,i}})$),
\[
z(w_{[j,j+2^{m+1}\ell)})\ge z(w_{J_{m,i}})=z_{I_{m,i}}>2\eps^{-1}\ell\cdot|I_{m,i}|=\eps^{-1}\cdot2^{m+1}\ell,
\]
so $j$ is a Blue $(2^{m+1}\ell)$-flag, and thus an $\ell^{+}$-flag
as desired since $\ell$ is a power of $2$. 
 
We conclude that all of the elements of
\begin{equation}
\label{eq:not-quite-disjoint}    
T\defeq[L]\cap\bigcup_{I_{m,i}\in\I}J_{m,i-1}
\end{equation}
are Blue $\ell^{+}$-flags of $w$. 

The union in \eqref{eq:not-quite-disjoint} may not be a disjoint union, but it is still large enough. To see this, say two intervals $I_{m_1, i_1}, I_{m_2, i_2} \in \I$ \emph{conflict} if their predecessor intervals $I_{m_1,i_1-1}, I_{m_2,i_2-1}$ have nonempty intersection. 
In this case, one of $I_{m_1,i_1-1},I_{m_2,i_2-1}$ must be a subset of the other as they are dyadic intervals.
We show we can find a large subcollection $\I'\subseteq \I$ of intervals no two of which conflict.
Choose $\I'$ greedily by repeatedly adding to $\I'$ the largest interval in $\I$ that does not conflict with any interval already in $\I'$.
If some interval $I_{m_1,i_1}\in\I$ is not added to $\I'$, it must conflict with some interval $I_{m_2,i_2}\in \I'$ with $m_2\ge m_1$ (in fact $m_2>m_1$). Then, by definition of conflict, $I_{m_2,i_2-1}$ is a superset of $I_{m_1,i_1-1}$, and thus also a superset of $I_{m_1,i_1}$ ($I_{m_1,i_1}$ is disjoint from $I_{m_2,i_2}$ by choice of $\I$).
Thus, the union of intervals in $\I\setminus \I'$ is a subset of the union of the predecessors of intervals in $\I'$.
As elements of $\I$ are pairwise disjoint, it follows that the total length of the intervals in $\I\setminus\I'$ is at most the total length of the intervals in $\I'$, and so the total size of the intervals in $\I'$ is at least half of the total size of intervals in $\I$.

Thus, $\sum_{\I'}|I_{m,i}|\ge\frac{1}{2}\sum_{\I}|I_{m,i}|$, and since no two intervals in $\I'$ conflict, the intervals $\{J_{m,i-1} \mid I_{m,i}\in\I'\}$ are pairwise disjoint. It is possible for $\I'$ to contain intervals $I_{m,i}$ with $i=0$, which would cause $J_{m,i-1}$ to lie outside $[L]$. However, the corresponding intervals $J_{m,i-1}$ are all pairwise disjoint, and if $i=0$ then $J_{m,i-1}$ contains $r-1$. Thus, there can be at most one interval $I_{m,i}$ for which $i=0$ and $J_{m,i-1}$ does not lie entirely inside $[L]$. This exceptional interval $J_{m,i-1}$ must have size
at most $\eps^{2}L$ or else
\[
z(w_{J_{m,i-1}}) > (1+\eps)\cdot 2^m, 
\]
for $2^m \ge \eps^2 L$, implying condition~\ref{enu:bp-skew}. It follows that once we exclude at most one exceptional interval, $w$ has at least
\begin{align}
|T|\ge\sum_{I_{m,i}\in\I'}|J_{m,i-1}|-\eps^{2}L\ge\frac{1}{2}\sum_{I_{m,i}\in\I}|J_{m,i-1}|-\eps^{2}L\ge\frac{(\alpha-10\eps)\eps}{8}L,
\label{eq:t-lower-power-two}
\end{align}
Blue $\ell^+$-flags, where the last step used \eqref{eq:length-of-maximal}. This proves the lemma, in fact with a stronger lower bound on number of Blue $\ell^+$ flags, when $\ell$ is a power of $2$.

Now suppose $\ell$ is not a power of $2$. If $\ell'$ is the smallest power of $2$ greater than or equal to $\ell$, then all Blue $(\ell')^+$-flags are also Blue $\ell^+$-flags and $1\le \ell' \le L$ since $L$ is a power of $2$. Furthermore, if $B_{\ell'}$ is the set of Blue $\ell'$-flags, $z(w_{[i,i+\ell)}) \le z(w_{[i,i+\ell')}) \le \eps^{-1} (\ell'-1)$ if $i\not\in B_{\ell'}$. Thus, since $|B_\ell|<\eps^2 L$ we get
\[
\frac{1}{\ell'}\sum_{i\in B_{\ell}\setminus B_{\ell'}}z(w_{[i,i+\ell)})\le \eps L,
\]
which implies
\[
\frac{1}{\ell'}\sum_{i\in B_{\ell'}}z(w_{[i,i+\ell')}) \ge \frac{1}{\ell'}\sum_{i\in B_{\ell}\cap B_{\ell'}}z(w_{[i,i+\ell)})\ge \frac{1}{2\ell}\sum_{i\in B_{\ell}}z(w_{[i,i+\ell)}) - \eps L \ge (\alpha/2 - \eps)L,
\]
where in the last inequality we used the original assumption from the lemma statement that $\frac{1}{\ell}\sum_{B_{\ell}}z(w_{[i,i+\ell)}) \ge \alpha L$.
Thus, the conditions of the lemma are true with modified parameters $\ell'\ge \ell$ a power of $2$ and $\alpha'=\alpha/2 - \eps$. It follows by applying (\ref{eq:t-lower-power-two}) with these parameters instead that $w$ must have at least
\[
\frac{(\alpha'-10\eps)\eps}{8}L\ge \frac{(\alpha - 22\eps)\eps}{16}L
\]
Blue $\ell^+$-flags in the general case, thus completing the proof.
\end{proof}

Now we are ready to prove Lemma~\ref{lem:classify}.

\begin{proof}[Proof of Lemma~\ref{lem:classify}.]
We assume conditions \ref{enu:skew} and \ref{enu:Green} do not
hold for some $w$, and prove condition \ref{enu:Blue-Red}.

Pick $m\in[0,n]$ maximal such that $w$ contains at least $\eps^{2}L$
Blue $(2^{m})^{+}$-flags. To see that such an $m$ exists, let $B_{1}$
denote the set of Blue $1$-flags in $w$, which is just the set
of one-bits in $w$ immediately followed by at least one zero. Then, we have
\[
\sum_{i\in B_{1}}z(w_{[i,i+1)})=z(w_{[L]})\ge L-\eps L,
\]
assuming condition \ref{enu:skew} does not hold. Furthermore, because
Blue flags are Green flags, condition \ref{enu:Green} being false
implies there are fewer than $\eps^{2}L$ Blue $2^{t}$-flags for
any $0\le t\le n$. In particular, $|B_{1}|<\eps^{2}L$. Thus, the
conditions of Lemma \ref{lem:blue-plus} are satisfied with $\ell = 1$
and $\alpha=1-\eps$, and we obtain that either $w$ satisfies condition~\ref{enu:skew} or condition~\ref{enu:bp-Green} (since Blue flags are Green flags) or the
number of Blue $1^{+}$-flags in $w$ is at least 
\[
(\alpha-22\eps)\eps L/16=(\eps-23\eps^{2})L/16\ge2\eps^{2}L.
\]
Here we used $\eps$ is sufficiently small. Thus, some such $0\le m\le n$ exists. 

By the maximality of $m$ and the fact that there are fewer than $\eps^{2}L$
Blue $2^{t}$-flags for any particular $0\le t\le n$, we see that in fact $m\ge 1$ and the number
of Blue $(2^{m})^{+}$-flags is in $[\eps^{2}L,2\eps^{2}L)$. Indeed, if the number of such flags were at least $2\eps^2L$, then at most $\eps^2L$ of these flags are Blue $2^m$-flags, so the number of Blue $(2^{m+1})^{+}$-flags would be at least $\eps^2 L$, contradicting the maximality of $m$. It remains
to check that the second half of condition \ref{enu:Blue-Red} holds.

Suppose $\ell\ge 2^m$, and let $B_{\ell},G_{\ell},Y_{\ell},$
and $R_{\ell}$ be the sets of $\ell$-flags in $w$ which are Blue,
Green, Yellow, and Red (respectively). Note that by our definitions
of the colors, $B_{\ell}\subseteq G_{\ell}\subseteq Y_{\ell}$ and
$Y_{\ell}\sqcup R_{\ell}=[L]$. Suppose for the sake of contradiction
that $|R_{\ell}|>600\eps L$. We may assume $\ell\le\eps^{2}L$, as otherwise a single Red $\ell$-flag would violate condition \ref{enu:skew}. We can express $z(w_{[L]})$ as
\[
z(w_{[L]})=\frac{1}{\ell}\sum_{i=-\ell}^{L}z(w_{[i,i+\ell)})
\]
since each $z(w_{\{i\}})$ appears exactly $\ell$ times in the sum on the
right. Setting aside the terms on the right with
$i\le0$, we find by breaking up the sum in terms of the colors of
the flags,
\begin{equation}
z(w_{[L]})\le\frac{1}{\ell}\left(\sum_{i=-\ell}^{0}z(w_{[i,i+\ell)})+\sum_{i\in B_{\ell}}z(w_{[i,i+\ell)})+(\ell-1)(\eps^{-1}|G_{\ell}\backslash B_{\ell}|+(1+2\eps)|Y_{\ell}\backslash G_{\ell}|+0.9|R_{\ell}|)\right).\label{eq:double-count}
\end{equation}
We can bound $\frac{1}{\ell}\sum_{i=-\ell}^{0}z(w_{[i,i+\ell)})\le z(w_{[\ell]})$,
$|G_{\ell}\backslash B_{\ell}|\le|G_{\ell}|<\eps^{2}L$ (since condition
\ref{enu:Green} is false), and $|Y_{\ell}\backslash G_{\ell}|\le|Y_{\ell}|=L-|R_{\ell}|$.
Putting these together with (\ref{eq:double-count}), we obtain

\begin{equation}
z(w_{[L]})\le L+3\eps L + z(w_{[\ell]})+\frac{1}{\ell}\sum_{i\in B_{\ell}}z(w_{[i,i+\ell)})-0.1|R_{\ell}|.\label{eq:simplified}
\end{equation}
Since condition \ref{enu:skew} is false, we have $z(w_{[L]})\ge L-\eps L$,
so together with (\ref{eq:simplified}) and $|R_\ell|> 600\eps L$ we get
\[
z(w_{[\ell]})+\frac{1}{\ell}\sum_{i\in B_{\ell}}z(w_{[i,i+\ell)})\ge56\eps L.
\]
Next, we claim that $z(w_{[\ell]})\le\eps L$. This follows from the
facts that $\ell\le\eps^{2}L$ and $z(w_{[\eps^{2}L]})\le\eps L$ (by the assumption that condition \ref{enu:skew} does not hold). We get
\[
\frac{1}{\ell}\sum_{i\in B_{\ell}}z(w_{[i,i+\ell)})\ge55\eps L.
\]
The conditions of Lemma
\ref{lem:blue-plus} are satisfied with this $\ell$ and $\alpha=55\eps$.
Applying the lemma, we find that either $w$ satisfies one of condition~\ref{enu:skew} or~\ref{enu:Green}, or the number of Blue $\ell^{+}$-flags
in $w$ is at least $(55\eps-22\eps)\eps L/16 > 2\eps^{2}L$.
This contradicts the observation we made before that the number of
Blue $(2^{m})^{+}$-flags is in $[\eps^{2}L,2\eps^{2}L)$ and completes
the proof.
\end{proof}

\subsection{Definition of types}

Using the structure theorem, we define below the \emph{type} of a string roughly based on the case that it satisfies in the structure theorem. 
These three definitions of types (Imbalanced, Green, Blue-Yellow) roughly align with the three cases of the structure theorem, though for Green and Blue-Yellow types, it helps to additionally have an upper bound on $\ell$, the length of the flags. 
Because we only prove the Structure Lemma (Lemma~\ref{lem:classify}) for strings whose number of ones is a power of two, we also only define types for strings whose number of ones is a power of two, which is enough for our application. Recall that $b_w(i)$ is defined to be the largest $\ell\in [L]$ such that $\ell$ is a power of two and $i$ is a Blue $\ell$-flag in $w$, with $b_w(i)=0$ if no such $\ell$ exists. Also, recall that $\gamma=0.001\varepsilon^2$.
\begin{definition}
  \label{def:type}
  Given a string $w$ with $L=2^n$ ones with $n$ sufficiently large, we say the \emph{type of $w$} is
  \begin{enumerate}
    \item {\bf Imbalanced} if there exists some interval $I\subseteq[L]$ of size $|I|\ge \varepsilon^{5} L$ such that $w_I$ is imbalanced.
    \item {\bf $\ell$-Green} for some integer $1\le \ell\le \varepsilon^{5}L$ if the number of Green $\ell$-flags in $w$ is at least $\eps^{2}L$. 
    \item {\bf $m$-Blue-Yellow} if there exists $1\le m\le n$ such that the number of indices $i\in[L]$ with $2^m\le b_w(i)\le \gamma L$ is at least $(\eps^{2}-\gamma)L$, and the number of Red $\ell$-flags in $w$ is at most $600\eps L$ for any $\ell \ge 2^m$.
  \end{enumerate}
  If $w$ could be more than one type, we assign $w$ one of the possible types arbitrarily.
\end{definition}
Note that there are at most $1+\eps^5 L+\log L= O(|w|)$ possible types for a string $w$. 
As a corollary of Lemma~\ref{lem:classify}, each string $w$ has a type, assuming $n$ is sufficiently large: 
\begin{lemma}
  \label{lem:type}
  If $n$ is sufficiently large, each string $w$ with $L=2^n$ ones has a type.
\end{lemma}
\begin{proof}
  If $w$ satisfies Case~\ref{enu:skew} of Lemma~\ref{lem:classify}, then $w$ has type Imbalanced.
  
  If $w$ satisfies Case~\ref{enu:Green} of Lemma~\ref{lem:classify} with parameter $\ell$, then $w$ has at least $\varepsilon^2L$ Green $\ell$-flags, and in particular there exists an $i\le L-\varepsilon^2L$ that is a Green $\ell$-flag.
  If $\ell\ge \varepsilon^{5}L$, then by Lemma~\ref{lem:blue-0}, the interval $I=[i,i+\ell-1]$ of size at least $\varepsilon^{5}L$ gives that $w_I$ is imbalanced, so $w$ has type Imbalanced.
  If $\ell\le \varepsilon^{5}L$, then $w$ has type $\ell$-Green.

  If $w$ satisfies Case~\ref{enu:Blue-Red} of Lemma~\ref{lem:classify} with parameter $m$, then there are at least $\varepsilon^2L$ indices with $b_w(i)\ge \varepsilon^2 L$, and the number of Red $\ell$-flags in $w$ is at most $600\varepsilon L$ for any integer $\ell\ge 2^m$.
  If there are at least $\gamma L$ indices $i\in[L]$ with $b_w(i)\ge \gamma L$, then some $i\le L-\gamma L + 1$ satisfies $b_w(i)\ge \gamma L$, so by Lemma~\ref{lem:blue-0}, there exists an interval $I$ of size at least $\gamma L > \varepsilon^5L$ such that $w_I$ is imbalanced, so $w$ is type Imbalanced.
  Otherwise, fewer than $\gamma L$ indices $i$ satisfy $b_w(i) \ge \gamma L$, so for the remaining at least $(\varepsilon^2-\gamma)L$ indices $i\in[L]$ satisfying $b_w(i) \ge \eps^2 L$, we have $2^m\le b_w(i)\le \gamma L$, and $w$ is type $m$-Blue-Yellow as desired.
\end{proof}

\section{The entropy regularity argument}\label{sec:entropy}

In this section, we prove the regularity-type result Lemma~\ref{lem:entropy}, which roughly states that in most dyadic substrings of a given string $s$, the positions of the Blue flags are distributed relatively uniformly.
In the Blue-Yellow strategy, we may be matching bits in $s$ and $t$ that lie in nearby but different dyadic intervals (because of our random shifting argument and because the Blue-Yellow strategy consumes ones from $s$ and $t$ at different rates). 
Because of this, it is helpful to say that there are many neighboring pairs of dyadic intervals with similar Blue flag distributions in $s$ and $t$.

\subsection{Flag balance}
Define the $L^1$ distance between two discrete probability distributions $p,q$ to be $\|p-q\|_{1}\defeq\sum|p(x)-q(x)|$.
Recall that $b_w(i)$ is the largest $\ell\in[L]$ such that $\ell$ is a power of two and $i$ is a Blue $\ell$-flag in $w$, and $0$ if no such $\ell$ exists. For a string $w$ with $L$ ones and an interval $I\subset[L]$, let $p_{w,I}$ denote the distribution of the value of $b_w(i)$ over a uniform random $i\in I$. 
Put another way, the probability mass $p_{w,I}(\ell)$ is the fraction of indices $i\in I$ with $b_w(i) = \ell$.

\begin{definition}[Blue-flag-balance]
\label{def:blue-flag-balance}
For $\beta>0$, we say that a dyadic interval $I_{m,i}$ is \emph{$\beta$-Blue-flag-balanced in $w$} if $\|p_{w,I_{m-1,2i-1}} -p_{w,I_{m-1,2i}}\|_1 \le \beta$.
We say that a string $w$ with $L=2^n$ ones is $\beta$\emph{-Blue-flag-balanced} if the interval $I_{n,1}$ is $\beta$-Blue-flag-balanced in $w$.
\end{definition}

Showing Blue-flag balance is useful because we can show that if a Blue-flag-balanced string $w$ has many Blue flags of a certain length, then both halves of $w$ also have many Blue flags of the same length.
The next lemma formalizes this idea. For a given set $U$ and property $P(i)$ depending on $i$, write $\Pr_{i\in U}[P(i)]$ for the probability that a uniform random element $i\in U$ satisfies $P(i)$.
\begin{lemma}
  \label{cor:flag-0}
  If string $w$ is $\beta$-Blue-flag-balanced with $L=2^n$ ones, then for any set $S$ of integers, we have
  \begin{equation*}
    \left|\Pr_{i\in [L]}[b_w(i)\in S] - \Pr_{i\in[L/2]}[b_w(i)\in S]\right| \le \frac{\beta}{2}.
  \end{equation*}
\end{lemma}
\begin{proof}
  Since $w$ is $\beta$-Blue-flag-balanced, interval $I_{n,1}$ is $\beta$-Blue-flag-balanced in $w$.
  We have
  \begin{align*}
    \left|\Pr_{i\in [L]}[b_w(i)\in S] - \Pr_{i\in[L/2]}[b_w(i)\in S]\right| 
    &= \sum_{\ell'\in S}^{} |p_{w,I_{n,1}}(\ell') - p_{w,I_{n-1,1}}(\ell')| \\
    &\le \sum_{\ell'}^{} |p_{w,I_{n,1}}(\ell') - p_{w,I_{n-1,1}}(\ell')| \nonumber\\
    &= \bigl\|p_{w,I_{n,1}} - p_{w,I_{n-1,1}}\bigr\|_1 \nonumber\\
    &= \Bigl\|\frac{p_{w,I_{n-1,1}} - p_{w,I_{n-1,2}}}{2}\Bigr\|_1
    \le \frac{\beta}{2} .
  \end{align*}
  The third equality uses that $\frac{p_{w,I_{n-1,1}} + p_{w,I_{n-1,2}}}{2} = p_{w,I_{n,1}}$, and the last inequality uses the definition of interval $I_{n,1}$ being $\beta$-Blue-flag-balanced in $w$.
\end{proof}

\subsection{Flag balance of intervals}
Our goal is to find a scale at which most intervals are Blue-flag-balanced.
We start by proving a simple lemma about probability distributions.
Recall that the binary entropy of a discrete probability distribution $p$ is defined as $H(p)\defeq-\sum_{x}p(x)\log p(x)$ over all values $x$ in the support of $p$, and the logarithms are base $2$. 
\begin{lemma}
\label{lem:pinsker}If $p^{-},p$, and $p^{+}$ are three discrete probability
distributions supported on a finite domain $\Omega$ satisfying $p^{-}(x)+p^{+}(x)=2p(x)$
for all $x\in \Omega$, then
\[
H(p^{-})+H(p^{+})\le2H(p)-\frac{1}{4}\|p^{+}-p^{-}\|_{1}^{2} . 
\]
\end{lemma}
\begin{proof}
We use Pinsker's inequality (see for example \cite[Section 2.8]{pinsker-ref}, which states in the case of discrete
probability distributions that the Kullback-Leibler divergence between
two distributions $P$ and $Q$ satisfies
\[
D_{\text{KL}}(P\|Q)=\sum_{i}P(i)\log\left(\frac{P(i)}{Q(i)}\right)\ge\frac{1}{2}\|P-Q\|_{1}^{2} .
\]
In particular, applying this to the pairs $(P,Q)=(p^{-},p)$ and $(P,Q)=(p^{+},p)$,
we obtain
\begin{align*}
2H(p)-H(p^{-})-H(p^{+}) & = \sum_{i}\left(p^{-}(i)\log p^{-}(i)+p^{+}(i)\log p^{+}(i)-2p(i)\log(p(i))\right)\\
 & = D_{\text{KL}}(p^{-}\|p)+D_{\text{KL}}(p^{+}\|p)\\
 & \ge  \frac{1}{2}\|p^{-}-p\|_{1}^{2}+\frac{1}{2}\|p^{+}-p\|_{1}^{2}\\
 & = \frac{1}{4}\|p^{+}-p^{-}\|_{1}^{2},
\end{align*}
where the last inequality follows from the fact that $p^+(x)-p^-(x) = 2(p(x) - p^-(x)) = 2(p^+(x) - p(x))$ for all $x\in \Omega$.
\end{proof}

We now obtain with a regularity argument the following lemma, which is the most substantial result in this section.
\begin{lemma}[Interval Blue-flag-balance]
\label{lem:entropy-0}
For $\beta >0$, $n \ge 2$, and any $w\in\{0,1\}^{N}$ with $L=2^{n}$ ones, except for at most $32\beta^{-3}\log n$ values of $m\in[0,n]$, the following holds: 
the number of dyadic intervals $I_{m,i}$ that are not $\beta$-Blue-flag-balanced in $w$ is less than  $\beta\cdot2^{n-m}$.
\end{lemma}
\begin{proof}
Consider the expression
\[
E_{m}\defeq2^{m}\cdot\sum_{i=1}^{2^{n-m}}H(p_{w,I_{m,i}}) .
\]
By the definition
of an $\beta$-Blue-flag-balanced interval, we obtain that whenever interval $I_{m,i}$
is not $\beta$-Blue-flag-balanced in $w$,
\[
\|p_{w,I_{m-1,2i-1}}-p_{w,I_{m-1,2i}}\|_{1}\ge \beta.
\]
Applying Lemma \ref{lem:pinsker} to the three distributions $p = p_{w,I_{m,i}}$, $p^{-} = p_{w,I_{m-1,2i-1}}$, and $p^{+} = p_{w,I_{m-1,2i}}$, we obtain
\[
2H(p_{w,I_{m,i}}) - H(p_{w,I_{m-1,2i-1}}) - H(p_{w,I_{m-1,2i}}) \ge \frac{1}{4} \beta^{2}
\]
whenever $I_{m,i}$ is not $\beta$-Blue-flag-balanced in $w$. Summing over the $t_{m}$ dyadic intervals $I_{m,i}$ that are not $\beta$-Blue-flag-balanced in $w$, we obtain
\[
E_{m-1}\le E_{m}-2^{m-1}\cdot t_{m}\cdot\frac{1}{4}\beta^{2} = E_{m}-2^{m-3}t_{m}\beta^{2}.
\]
Since $b_w(i)\le L$ is either 0 or a power of two for all indices $i\in[L]$, we have that $b_w(i)$ takes on one of $n+2$ values.
Thus, we have $E_{n}=2^{n}\cdot H(p_{w,I_{n,1}})\le\log(n+2)\cdot2^{n}$.
Since we also have $E_0\ge 0$, we obtain from the previous equation that
\[
\sum_{m = 1}^{n} 2^{m-3}t_m \beta^2 \le \sum_{m = 1}^{n} E_m - E_{m-1} \le \log(n+2) \cdot 2^n.
\]
In particular, at most $8\beta^{-3}\log(n+2)<32\beta^{-3}\log n$ values of $t_{m}$ are at least $\beta\cdot2^{n-m}$, completing the proof.
\end{proof}

\subsection{Flag balance of substrings}
In Lemma~\ref{lem:entropy-0} we showed that there exists many scales $m$ where many (dyadic) \emph{intervals} were $\beta$-Blue-flag-balanced in the sense of Definition~\ref{def:blue-flag-balance}. 
For technical reasons, it is helpful to establish that there are many \emph{substrings} that are $\beta$-Blue-flag-balanced, and we do so in this section. The distinction here is that certain one-bits in the interval $I$ may be Blue $\ell$-flags in $w$, but not Blue $\ell$-flags in $w_I$ once the zeros to the right of $I$ are taken out of consideration.

For a string $w$ with $L$ ones, let $p_{w}\defeq p_{w,[L]}$.
Note that the distribution $p_{w,I}$ may be different from the distribution $p_{w_I}$, because indices $i\in I$ that are Blue $\ell$-flags in string $w$ may not correspond to Blue $\ell$-flags in substring $w_I$.
However, the converse is true: for $I=[x,y]$, if $i-x+1$ is a Blue $\ell$-flags of substring $w_I$, then $i$ is a Blue $\ell$-flag of string $w$.
Because of this, we can show that $p_{w,I}$ and $p_{w_I}$ are similar in distribution under certain conditions.
\begin{lemma}
  \label{lem:flag-1}
  Let $\beta > 0$, $w$ be a string with $L$ ones, $I\subset[L]$ be an interval, and $\ell_0$ be an integer, such that at most $\beta|I|$ indices $i\in I$ satisfy $b_w(i)\ge \ell_0$.
  Then 
  \[ \|p_{w,I}-p_{w_I}\|_1\le 2\left(\beta+\frac{\ell_0}{|I|}\right). \]
\end{lemma}
\begin{proof}
  Recall that $b_w(i)$ is the largest power of two such that the $i$-th one of $w$ is an $b_w(i)$-flag of $w$ (or 0 if no such power of two exists). 
  For an interval $I=[x,y]$, if index $i-x+1$ is a Blue $\ell$-flag of $w_I$, then index $i$ is a Blue $\ell$-flag of $w$, and furthermore if index $i$ is a Blue $\ell$-flag of $w$ and $x\le i\le y-\ell$, then index $i-x+1$ is a Blue $\ell$-flag of $w_I$ as well.
  Hence, for all $x\le i\le y-\ell_0$ with $b_w(i)\le \ell_0$, we have that $b_{w}(i)= b_{w_I}(i-x+1)$.
  Thus, by the union bound,
  \begin{align}
    \Pr_{i\in I}[b_w(i)\neq b_{w_I}(i-x+1)] \le \Pr_{i\in I}[i\ge y-\ell_0] + \Pr_{i\in I}[b_w(i)\ge \ell_0] \le \frac{\ell_0}{|I|} + \beta,
    \label{eq:flag-1-1}
  \end{align}
  We also have 
  \begin{align}
    \|p_{w,I}-p_{w_I}\|_1
    = \sum_{\ell'}^{} |p_{w,I}(\ell')-p_{w_I}(\ell')| 
    = \sum_{\ell'}^{} \frac{1}{|I|}\abs{\sum_{i\in I}^{} \indicator(b_w(i)=\ell')-\sum_{i\in I}^{} \indicator(b_{w_I}(i-x+1)=\ell')} 
    \label{eq:flag-1-1b}
  \end{align}
  Changing the value of any single $b_w(i)$ changes the value of at most two indicator functions in \eqref{eq:flag-1-1b}, and furthermore changing $|I|\cdot \Pr_{i\in I}[b_w(i)\neq b_{w_I}(i-x+1)]$ values of $b_w(i)$ makes the expression $\|p_{w,I}-p_{w_I}\|_1$ equal to 0, so we have that
  \begin{align}
    \|p_{w,I}-p_{w_I}\|_1\le \frac{1}{|I|}\cdot 2\cdot |I|\cdot \Pr_{i\in I}[b_w(i)\neq b_{w_I}(i-x+1)]
    \le 2\cdot \Pr_{i\in I}[b_w(i)\neq b_{w_I}(i-x+1)].
    \label{eq:flag-1-2}
  \end{align}
  Combining \eqref{eq:flag-1-1} and \eqref{eq:flag-1-2} gives the desired result.
\end{proof}

Lemma~\ref{lem:entropy-0} argues about the Blue-flag-balance of intervals. 
Combining Lemma~\ref{lem:entropy-0} with Lemma~\ref{lem:flag-1}, we obtain the following result about the Blue-flag-balance of substrings.

\begin{lemma}[Substring Blue-flag-balance]
\label{lem:entropy}
Let $\beta>0$, and $n$ be sufficiently large in terms of $\beta$.
Let $w\in\{0,1\}^{N}$ be a string with $L=2^{n}$ ones, and suppose that at most $\beta^2 L$ indices $i\in[L]$ satisfy $b_w(i)\ge 2^{n-200\beta^{-3}\log n}$. 
  Then, except for at most $32\beta^{-3}\log n$ values of $m\in[n-150\beta^{-3}\log n,n]$ the following holds: the number of $i\in[2^{n-m}]$ for which $w_{m,i}$ is not $6\beta$-Blue-flag-balanced is less than $3\beta\cdot2^{n-m}$.
  
\end{lemma}
\begin{proof}
  Let $m\in[n-150\beta^{-3}\log n,n]$ be any value such that there are at most $\beta \cdot 2^{n-m}$ dyadic intervals $I_{m,i}$ that are not $\beta$-Blue-flag-balanced in $w$.
  By Lemma~\ref{lem:entropy-0}, all but $32\beta^{-3}\log n$ values of $m$ have this property. 
  We show that each such $m$ satisfy the requirements of the lemma.
  
  Call an index $i\in[2^{n-m}]$ \emph{good} if (1) the dyadic interval $I_{m,i}$ is $\beta$-Blue-flag-balanced in $w$, and (2) for both intervals $I\in\{I_{m-1,2i-1},I_{m-1,2i}\}$, there are at most $\beta|I|$ indices $i\in I$ with $b_w(i)\ge 2^{n-200\beta^{-3}\log n}$.
  By the choice of $m$, there are less than $\beta \cdot 2^{n-m}$ choices of $i\in[2^{n-m}]$ that violate requirement (1).
  Next, we use the assumption that in total at most $\beta^2 L$ indices $i\in[L]$ satisfy $b_w(i)\ge 2^{n-200\beta^{-3}\log n}$. Thus, at most $\beta \cdot 2^{n-m+1}$ of the dyadic intervals $I$ of size $2^{m-1}$ contain at least $\beta |I|$ indices $i\in I$ satisfying $b_w(i)\ge 2^{n-200\beta^{-3}\log n}$.
  Hence, at most $2\beta\cdot 2^{n-m}$ choices of $i\in[2^{n-m}]$ violate requirement (2).
  We see that all but less than $3\beta \cdot 2^{n-m}$ choices of $i\in [2^{n-m}]$ are good.

  We now show that for each good index $i\in[2^{n-m}]$, the substring $w_{m,i}$ is $6\beta$-Blue-flag-balanced.
  By Lemma~\ref{lem:flag-1}, for each good $i$ and either interval $I\in\{I_{m-1,2i-1},I_{m-1,2i}\}$, we have 
  \[
    \|p_{w,I}-p_{w_I}\|_1\le 2\left(\beta + \frac{2^{n-200\beta^{-3}\log n}}{2^m}\right)\le 2\left(\beta + 2^{-50\beta^{-3}\log n}\right).
  \]
  Since interval $I_{m,i}$ is $\beta$-Blue-flag-balanced with respect to $w$, we have by the triangle inequality
  \begin{align*}
    \|p_{w_{m-1,2i-1}}-p_{w_{m-1,2i}}\|_1
    &\le 
    \|p_{w_{m-1,2i-1}}-p_{w,I_{m-1,2i-1}}\|_1
    +\|p_{w_{m-1,2i}}-p_{w,I_{m-1,2i}}\|_1
    +\|p_{w,I_{m-1,2i-1}}-p_{w,I_{m-1,2i}}\|_1 \nonumber\\
    &\le
    4(\beta + 2^{-50\beta^{-3}\log n})
    +\beta
     < 6\beta,
  \end{align*}
  assuming $n$ is sufficiently large.
  Hence, for each good $i$, the substring $w_{m,i}$ is $6\beta$-Blue-flag-balanced, as desired.
\end{proof}

\section{Green case}\label{sec:green}

Call a pair of strings $(s,t)$ a \emph{Green pair} if
\begin{enumerate}
    \item  $s$ and $t$ have the same number $L$ of ones, and 
    \item there exists an $1\le \ell\le \varepsilon^5L$ such that $s$ and $t$ are both type $\ell$-Green.
\end{enumerate}
In this section, we implement the Green strategy in the Overview (Section~\ref{sec:overview}).
Specifically, we show (Lemma~\ref{lem:green2}) that when we have strings $s$ and $t$ and some scale $m^*$ with many Green pairs $(s_{m^*,i},t_{m^*,i})$, then we can find a common subsequence of $s$ and $t$ of length $(0.5+\delta)|s|$. 
At the highest level, we do this by finding common subsequences within the Green pairs, and matching ones elsewhere, using the Prefix/Suffix LCS Lemma (Lemma~\ref{lem:shift}).

Within the Green pairs, we match ones everywhere, except that we switch to matching zeros at Green flags in $s$ and $t$.

\begin{lemma}[Green matching lemma]
\label{lem:green}
  Let $L$ be a power of two.
  Let $(s,t)$ be a Green pair where strings $s$ and $t$ have $L$ ones each.
  Then for $\Delta$ uniformly random in $[-L,L]$, we have 
  \[
    \E_{\Delta}\left[\LCS(\Drop_\Delta(s), \Drop_{-\Delta}(t))+|\Delta|\right]\ge L+ \frac{\varepsilon^5}{8} L.
  \]
\end{lemma}
\begin{proof}
  Let $\ell$ be the integer that allows $(s,t)$ to be a Green pair.
  For a fixed $\Delta$ in $[-L,L]$, let $G_\Delta'$ be the set of indices $i\in[L]$ such that $i$ is Green $\ell$-flag of $s$ and $i+\Delta\in [L]$ and $i+\Delta$ is a Green $\ell$-flag of $t$.
  Let $G_\Delta\subseteq G_\Delta'$ be a subset of size at least $|G_{\Delta}'|/\ell -1$ such that any two distinct $i,i'\in G_\Delta$ satisfy $|i-i'|\ge \ell$, and furthermore every $i\in G_\Delta$ satisfies $i+\Delta+\ell\le L$.
  Such a $G_\Delta$ can be greedily selected from $G'_\Delta$: we can get $|G_\Delta|/\ell$ indices such that any two differ by at least $\ell$, and, as any $i\in G'_\Delta$ satisfies $i+\Delta\le L$, we can guarantee $i+\Delta+\ell\le L$ for all $i\in G_\Delta$ by removing the largest index.

  Since $s$ and $t$ are type $\ell$-Green, both $s$ and $t$ have at least $\varepsilon^2 L$ Green $\ell$-flags.
  Thus, for each Green $\ell$-flag $i\in[L]$ of $s$, there are at least $\varepsilon^2L$ values of $\Delta$ such that $i\in G_\Delta'$.
  Hence, $\Pr_{\Delta\in[-L,L]}[i\in G_\Delta']\ge \frac{\varepsilon^2L}{2L+1} > \frac{\varepsilon^2}{3}$.
  By linearity of expectation, we have $\E_{\Delta\sim[-L,L]}[|G_\Delta'|]\ge \frac{\varepsilon^2}{3}\cdot \varepsilon^2 L = \frac{\varepsilon^4}{3}L$. 
  Hence, as $|G_\Delta|\ge |G_\Delta'|/\ell-1$, we have
  \begin{align}
    \label{eq:green-1}
    \E_{\Delta\sim[-L,L]}[|G_\Delta|]
    \ge \frac{\varepsilon^4L}{3\ell} - 1 
    > \frac{\varepsilon^4L}{4\ell}.
  \end{align}
  Here we used that $\ell\le \varepsilon^{5}L$.

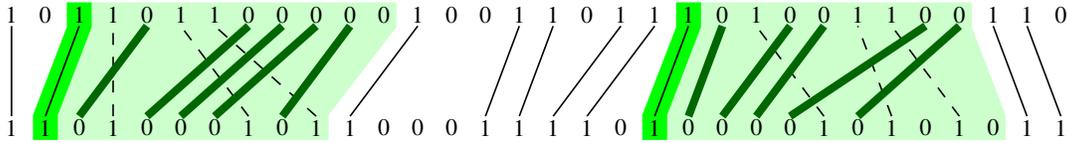
\begin{figure}[h]
\centering
\vspace{-5mm}
\scalebox{1.5}{
\begin{tikzpicture}
\drawpartg{0}{1}{1}{1}{0}{2}{1}{3}
\drawpartgg{0}{1}{0}{1}{0}{2}{0}{2}
\drawpartg{2}{3}{3}{5}{2}{4}{3}{4}
\drawpartgg{2}{3}{2}{3}{2}{4}{2}{4}
\subseqpic
\drawmatch{0,0}{0,0}
\drawmatch{0,1}{0,2}

\drawmatchlight{0,3}{0,3}
\drawmatchlight{0,7}{0,5}
\drawmatchlight{1,1}{0,6}
\drawmatchgreen{0,2}{0,4}
\drawmatchgreen{0,4}{0,7}
\drawmatchgreen{0,5}{1,0}
\drawmatchgreen{0,6}{1,1}
\drawmatchgreen{1,0}{1,2}
\drawmatch{1,2}{1,4}
\drawmatch{1,6}{1,7}
\drawmatch{1,7}{2,0}
\drawmatch{2,0}{2,2}
\drawmatch{2,1}{2,3}
\drawmatch{2,3}{2,4}
\drawmatchlight{3,0}{2,6}
\drawmatchlight{3,2}{3,1}
\drawmatchlight{3,4}{3,2}
\drawmatchgreen{2,4}{2,5}
\drawmatchgreen{2,5}{2,7}
\drawmatchgreen{2,6}{3,0}
\drawmatchgreen{2,7}{3,3}
\drawmatchgreen{3,1}{3,4}
\drawmatch{3,6}{3,5}
\drawmatch{3,7}{3,6}
\end{tikzpicture}
}
\vspace{-5mm}
\caption{The matching strategy for Lemma~\ref{lem:green}.
In this example $\ell = 4$ and $\Delta=0$.
The darker Green one-bits are Green $\ell$-flags, and the light Green substrings following them are relatively zero-rich. 
The solid and dashed black lines indicate the initial matching of ones, and the thick dark-green lines indicate matchings of zeros that replace the matchings of ones. The matchings of ones that are replaced are dashed.
}
\label{fig:green-case-2}
\end{figure}

  It suffices to construct a large matching between the bits of $s$ and the bits of $t$ such that only equal bits are matched and such that the matching is ``non-crossing'', meaning that earlier bits of $s$ are matched with earlier bits of $t$.
  Indeed, the number of pairs in a non-crossing matching corresponds to the length of a common subsequence of $s$ and $t$.

  Consider the matching where we match the $i$-th one in $s$ with the $(i+\Delta)$-th one in $t$.
  This matches $L-|\Delta|$ ones.
  In particular, for each $i\in G_\Delta$, the ones in substring $s_{[i+1,i+\ell-1]}$ are exactly matched with the ones in substring $t_{[i+\Delta+1,i+\Delta+\ell-1]}$.
  In this matching, for each $i\in G_\Delta$, replace the matching between the ones of $s_{[i+1,i+\ell-1]}$ and the ones of $t_{[i+\Delta+1,i+\Delta+\ell-1]}$ with a matching between the zeros of $s_{[i,i+\ell-1]}$ and the zeros of $t_{[i+\Delta,i+\Delta+\ell-1]}$.
  All of the zeros of $s_{[i,i+\ell-1]}$ come after the $i$-th one of $s$, and all the zeros of $t_{[i+\Delta,i+\Delta+\ell-1]}$ come after the $(i+\Delta)$-th one of $t$, which stays matched to the $i$-th one of $s$, so the matching stays non-crossing after this replacement. 
  Furthermore, since any two $i\in G_\Delta$ differ by at least $\ell$, the substrings $s_{[i,i+\ell-1]}$ and $t_{[i+\Delta,i+\Delta+\ell-1]}$ across $i\in G_\Delta$ are pairwise disjoint, so these $|G_\Delta|$ replacements can be done simultaneously while keeping the matching non-crossing and thus valid.

  Substrings $s_{[i,i+\ell-1]}$ and $t_{[i+\Delta,i+\Delta+\ell-1]}$ each have at least $\max(1,(1+\varepsilon)(\ell-1))$ zeros because $i$ is a Green $\ell$-flag of $s$ and $i+\Delta$ is a Green $\ell$-flag of $t$.
  Thus, each replacement deletes $(\ell-1)$ pairs of ones and adds at least $\max(1,(1+\varepsilon)(\ell-1))$ pairs of zeros.
  If $\ell\ge 2$, each replacement increases the number of pairs matched by $\varepsilon(\ell-1) > \frac{\varepsilon}{2}\ell$, so the total number of bits in the new matching is at least  $L-|\Delta| + \frac{\varepsilon}{2}\ell|G_\Delta|$.
  If $\ell=1$, each replacement increases the number of pairs matched by at least 1, so the total number of pairs in the new matching is $L-|\Delta| + |G_\Delta|$ which is also at least $L-|\Delta| + \frac{\varepsilon}{2}\ell|G_\Delta|$.
  Thus, for all $\Delta$, we have shown 
  \begin{align*}
    \LCS(\Drop_\Delta(s), \Drop_{-\Delta}(t)) \ge L-|\Delta|+\frac{\varepsilon}{2}\ell|G_\Delta|.
  \end{align*}

  Hence, we have by linearity of expectation and \eqref{eq:green-1},
  \begin{align*}
    \E_{\Delta\sim[-L,L]}\left[\LCS(\Drop_\Delta(s), \Drop_{-\Delta}(t)) + |\Delta| \right]
    \ge L + \frac{\varepsilon}{2}\ell\cdot \E_{\Delta\sim[-L,L]}\left[|G_\Delta|\right]
    \ge L + \frac{\varepsilon^5}{8}L,
  \end{align*}
  as desired.
\end{proof}

To find strings $s$ and $t$ that have LCS beating the 1/2 barrier, it is not enough to assume that $s$ and $t$ form a Green pair and use Lemma~\ref{lem:green}, because we lose by the size of the random shift $|\Delta|$. To remedy this, we break $s$ and $t$ into shorter substrings so that the loss from the random shift is at most the length of a single substring. If a substantial fraction of these substrings form Green pairs, then we can combine the gains using the Shifted LCS Lemma (Lemma~\ref{lem:shift}) to get a large overall LCS, and the loss from the random shift is negligible.
The following lemma implements this idea, showing that many Green pairs gives large overall LCS.

\begin{lemma}[Many Green pairs implies large LCS]
  \label{lem:green2}
  Let $s$ and $t$ be strings with the same length and the same number $L=2^n$ of ones, and let $m^*\le n-10-\log\varepsilon^{-5}$.
  Suppose there exists a set $Z\subset[2^{n-m^*}]$ such that $|Z| > 2^{n-m^*}/10$ and for all $i\in Z$, the pair of substrings $(s_{m^*,i},t_{m^*,i})$ is a Green pair.
  Then 
  \[
    \LCS(s,t) \ge \left( 0.5 + \frac{\varepsilon^5}{5000} \right)|s|.
  \]
\end{lemma}
\begin{proof}
Let $L^* = 2^{m^*}$.
We may assume $L\ge (0.5-\frac{\varepsilon^5}{1000})|s|$ or else we are done by having an LCS of $(0.5+\frac{\varepsilon^5}{1000})|s|$ zeros.
By Lemma~\ref{lem:green}, we have, for all $i\in Z$,
  \[
    \E_{\Delta\sim[-L^*,L^*]}\left[ \LCS\left(\Drop_{\Delta}(s_{m^*,i}), \Drop_{-\Delta}(t_{m^*,i})\right) + |\Delta|\right] \ge L^* + \frac{\varepsilon^5}{8}L^* .
  \]
  We have, by linearity of expectation,
  \begin{align}
    \E_{\Delta\sim[-L^*,L^*]}&\left[\sum_{i\in Z}^{} \left(\LCS\left(\Drop_{\Delta}(s_{m^*,i}), \Drop_{-\Delta}(t_{m^*,i})\right) + |\Delta|\right)\right] \nonumber\\
    &= \sum_{i\in Z}\E_{\Delta\sim[-L^*,L^*]}\left[ \LCS\left(\Drop_{\Delta}(s_{m^*,i}), \Drop_{-\Delta}(t_{m^*,i})\right) + |\Delta|\right] \nonumber\\
    &\ge |Z|\cdot \left(L^*+\frac{\varepsilon^5}{8}L^*\right) \ . \nonumber 
  \end{align}
  Hence, we may fix a single $\Delta$ for which $|\Delta| \le L^*$ and
  \[
   \sum_{i\in Z}^{} \LCS\left(\Drop_{\Delta}(s_{m^*,i}), \Drop_{-\Delta}(t_{m^*,i})\right)
    \ge |Z|\cdot \left(L^*-|\Delta|+\frac{\varepsilon^5}{8}L^*\right).
  \]
   \noindent Thus, the set $Z$ satisfies the setup of Lemma~\ref{lem:shift} with $n'=n$ and $m'=m^*$ and $L'=L$ and $\delta' = \frac{\varepsilon^5}{8}$. Therefore we get
  \[
    \LCS(s,t) \ge 
    \left( 1 + \frac{\varepsilon^5}{160} \right)L 
    > \left( 1 + \frac{\varepsilon^5}{160} \right)\left(0.5 - \frac{\varepsilon^5}{1000}\right)|s|
    > \left( 0.5+ \frac{\varepsilon^5}{5000} \right)|s|. \qedhere
  \]
\end{proof}

\section{Blue-Yellow case}\label{sec:blue}
In this section, we implement the Blue-Yellow strategy described in the Overview (Section~\ref{sec:overview}), which is the most intricate part of our argument.
Call a pair of strings $(s,t)$ a \emph{Blue-Yellow pair} if: 
\begin{enumerate}
    \item The strings $s$ and $t$ have the same number of ones.
    \item There exists an $m$ such that $s$ and $t$ are both of type $m$-Blue-Yellow.
    \item Both $s$ and $t$ are $6\gamma$-Blue-flag-balanced.
\end{enumerate}
Note that $(s,t)$ is a Blue-Yellow pair if and only if $(t,s)$ is a Blue-Yellow pair.

We show that when we have strings $s$ and $t$ and some scale $m^*$ with many Blue-Yellow pairs $(s_{m^*,i},t_{m^*,i})$ among their substrings at scale $m^*$, then we can find a common subsequence of $s$ and $t$ of length $(1/2+\delta)|s|$. 
At the highest level (Lemma~\ref{lem:blue3}), we do this by finding common subsequences within the Blue-Yellow pairs, and matching ones elsewhere.
To find large common subsequences within Blue-Yellow pairs, we use Lemma~\ref{lem:blue1} and Lemma~\ref{lem:blue2}.

Lemma~\ref{lem:blue1} shows that, for a Blue-Yellow pair $(s_{m^*,i},t_{m^*,i})$, we can find substrings $s_{m^*,i}'$ and $t_{m^*,i}'$ of $s_{m^*,i}$ and $t_{m^*,i}$ which ``gain in span,'' meaning that $\LCS(s_{m^*,i}', t_{m^*,i}') > (1/4+\delta)(|s_{m^*,i}'| + |t_{m^*,i}'|)$.
To do this, we match ones of $s_{m^*,i}$ with ones of $t_{m^*,i}$, and switch to matching zeros when we simultaneously encounter a Blue flag in $s_{m^*,i}$ and Yellow flag in $t_{m^*,i}$ of the appropriate lengths.
However, despite gaining in span, the lengths of $s_{m^*,i}$ and $t_{m^*,i}$ may be quite different (this offset comes from the zeros of the Blue flags in $s_{m^*,i}$ spanning fewer ones than the zeros of the Yellow flags in $t_{m^*,i}$), so repeatedly applying Lemma~\ref{lem:blue1} is insufficient. Indeed, this would cause us to systematically match shorter substrings in $s$ with longer substrings in $t$, leading the common subsequence obtained to reach the end of $t$ well before reaching the end of $s$.

The subsequent lemma, Lemma~\ref{lem:blue2}, shows that if \emph{two} Blue-Yellow pairs $(s_{m^*,i},t_{m^*,i})$ and $(s_{m^*,i+1},t_{m^*,i+1})$ occur consecutively in a string, then we can find a common subsequence that gains in span, like Lemma~\ref{lem:blue1}, but also uses the same number of ones in both strings, creating a balanced matching.
Roughly, the proof follows by applying Lemma~\ref{lem:blue1} twice, once on the pair $(s_{m^*,i},t_{m^*,i})$, matching more bits in $t_{m^*,i}$, and once on the pair $(t_{m^*,i+1},s_{m^*,i+1})$, matching more bits in $s_{m^*,i+1}$, so that together the number of bits used from $s$ and $t$ is equal. As a result, in total our common subsequence spans the same number of bits in $s$ and $t$ but still gains in span. We refer the reader to Figure~\ref{fig:blueyellow-case-2} below for a visual illustration of this argument.

Thanks to Lemma~\ref{lem:blue2}, it follows that if we have many Blue-Yellow pairs, then we have many pairs of subsequences $(s_{m^*,i}s_{m^*,i+1},t_{m^*,i}t_{m^*,i+1})$ of $s$ and $t$ where there is a common subsequence that gains in span.
Thus, using Lemma~\ref{lem:shift}, we can piece these subsequences together, matching ones in between, giving a large $\LCS$ overall; this is the content of Lemma~\ref{lem:blue3}.
\begin{lemma}[Blue-Yellow matching lemma]
  \label{lem:blue1}
  Let $(s,t)$ be a Blue-Yellow pair where each string has $L=2^n$ ones.
  Then, for $\Delta$ uniformly random in $[L/4]$, 
  \[
    \Pr_{\Delta}\left[\LCS(s_{[1,(0.5+0.01\varepsilon)L]}, t_{[1+\Delta, (0.5+0.3\varepsilon)L + \Delta]}) \ge (0.5 + 0.24 \varepsilon)L\right]\ge 0.9.
  \]
\end{lemma}
\begin{proof} This proof is long, so we organize it into four parts. First, we use the assumptions to find many disjoint Blue flags in $s$. Next, we match the Blue flags in $s$ with candidate Yellow flags in $t$. We then show that, on average, most of these candidate flags are in fact Yellow flags in $t$. Finally, with these ingredients in place, we show the desired lower bound on the expected LCS over the random offset $\Delta$.

\medskip\noindent\emph{Finding many disjoint Blue flags in $s$.}
  Since $(s,t)$ is a Blue-Yellow pair, the strings $s$ and $t$ are both type $m$-Blue-Yellow for some integer $m$.
  Because $s$ is type $m$-Blue-Yellow, the fraction of indices $i\in[L]$ with $b_s(i)\in [2^m,\gamma L]$ is at least $\varepsilon^2-\gamma$. 
  As $s$ is $6\gamma$-Blue-flag-balanced, by Lemma~\ref{cor:flag-0}, the fraction of indices $i\in[L/2]$ with $b_s(i)\in[2^m,\gamma L]$ is at least $\varepsilon^2-4\gamma$.
  Let $1\le i_1<i_2<\cdots$ be such that $i_k\in[L/2]$ is the smallest index satisfying $b_s(i_k)\in[2^m,\gamma L]$ and, if $k > 1$, $i_k \ge i_{k-1}+b_s(i_{k-1})$. Intuitively, $(i_k)_{k\ge 1}$ is a maximal sequence of non-overlapping Blue flags in $[L/2]$.

  We claim that there is a smallest index $d$ such that $\sum_{k=1}^{d} b_s(i_k)\ge 0.5(\varepsilon^2-4\gamma)L$.
  By the minimality of each $i_k$, the intervals $[i_k,i_k+b_s(i_k)-1]$ for $k=1,\dots,k'$ together contain all indices $i < i_{k'+1}$ with $b_s(i)\in[2^m,\gamma L]$.
  Since there are at least $0.5(\varepsilon^2-4\gamma)L$ such indices $i \in [L/2]$, as long as $\sum_{k=1}^{k'} b_s(i_{k})<0.5(\varepsilon^2-4\gamma)L$, the index $i_{k'+1}$ is well defined, so in particular $d$ is well defined.
  Furthermore, since $d$ is minimal and $b_s(i_k)\le \gamma L$ for all $i_k$, we have $\sum_{k=1}^{d} b_s(i_k)\le 0.5(\varepsilon^2-4\gamma)L + \gamma L < 0.5\varepsilon^2L$.
  We thus have found indices $i_1,\dots,i_d$ where 
  \[ i_d\le L/2, \ i_{k+1}\ge i_k+b_s(i_k) \ \text{for each} \ k<d, \ \text{and} \ \sum_{k=1}^d b_s(i_k)\in[0.5(\varepsilon^2-4\gamma)L, 0.5\varepsilon^2L] . \]

\smallskip
\noindent\emph{Matching Blue flags in $s$ with candidate Yellow flags in $t$.}  
  Define $i_0=1$, $\ell_{s,0} = 0$, and $i_{d+1}=\floor{(0.5+0.01\varepsilon)L}$.
  For $k\in[d]$, let $\ell_{s,k}\defeq b_s(i_k)>1$ for the lengths of these Blue flags in $s$, and let $\ell_{t,k} \defeq \ceil{0.56\varepsilon^{-1}\ell_{s,k}}$ for $k\in[0,d]$, corresponding to Yellow flag length we want to match in $t$. Here the constant $0.56 \eps^{-1}$ above is chosen so that $0.9(\ell_{t,k}-1) \ge 0.5 \eps^{-1} \ell_{s,k}$, so that a Yellow $\ell_{t,k}$-flag guarantees roughly the same number of zeros in $t$ as a Blue $\ell_{s,k}$-flag does in $s$.
  
  For a subset $K\subset[0,d]$, define $\ell_{s,K} \defeq \sum_{k\in K}^{} \ell_{s,k}$, and define $\ell_{t,K}$ similarly.
  Recall the shift $\Delta$ is to be chosen uniformly from $[L/4]$.
  For $k\in [0, d+1]$, let $j_k\defeq i_k+\Delta+\ell_{t,[k-1]}-\ell_{s,[k-1]}$. This $j_k$ is the index in $t$ that we would like to be a Yellow $\ell_{t,k}$-flag to be matched with the Blue $\ell_{s,k}$-flag at $i_k$. 
  The indices $i_k$ and $j_k$ partition $s$ and $t$ into substrings as follows. Let
  \begin{align*}
    s_k'\ &\defeq \ s_{[i_k,i_k+\ell_{s,k}-1]} &&\text{for $k\in [1,d]$} \\
    s_k''\ &\defeq \  s_{[i_k+\ell_{s,k}, i_{k+1}-1]} &&\text{for $k\in[0,d]$} \\
    t_k'\ &\defeq \ t_{[j_k,j_k+\ell_{t,k}-1]} &&\text{for $k\in [1,d]$} \\
    t_k''\ &\defeq \  t_{[j_k+\ell_{t,k}, j_{k+1}-1]} &&\text{for $k\in[0,d]$.}
  \end{align*}
  By construction of $i_0,\dots,i_d$, we have $i_{k} + \ell_{s,k}\le i_{k+1}$ for $k\in [0, d-1]$, and also $i_d+\ell_{s,d} \le L/2 + \gamma L < i_{d+1}$. 
  Thus, the substrings defined above give a partition $s_{[1,(0.5+0.01\varepsilon)L]}=s_0''s_1's_1''s_2's_2''\dots s_d's_d''$, where the substrings alternate between the regions $s_k'$ corresponding to Blue flags and the regions $s_k''$ in between.
  Similarly, for $k\in [0,d-1]$ we have
  \[
  j_k + \ell_{t,k}
  = i_k + \Delta + \ell_{t,[k]} - \ell_{s,[k-1]}
  = j_{k+1} - (i_{k+1}-i_k-\ell_{s,k})
  \le j_{k+1},
  \]
  and
  \begin{align}
    j_{d+1}
    &= i_{d+1} +  \ell_{t,[d]}-\ell_{s,[d]} +\Delta \nonumber\\
    &\le 1+L/2 + \gamma L + 0.57\varepsilon^{-1}\cdot\ell_{s,[d]} + \Delta \nonumber\\
    &\le 1+L/2 + 0.01\varepsilon L +0.57\varepsilon^{-1}\cdot 0.5\varepsilon^2 L + \Delta \nonumber\\
    &\le 1+(0.5 + 0.3\varepsilon)L +\Delta. \nonumber
  \end{align} 
  Thus, we have that $t_0''t_1't_1''t_2't_2''\dots t_d't_d''$ is a prefix of $t_{[1+\Delta, (0.5+0.3\varepsilon)L+\Delta]}$, alternating between the regions $t_k'$ that we wish to cover by Yellow flags and the regions $t_k''$ in between.

\medskip\noindent\emph{Showing many $j_k$'s are Yellow flags of $t$.}
  We next show that typically, most of the $j_k$'s as defined above are Yellow flags. Let $K$ be the set of $k\in[d]$ for which the $j_k$ is a Yellow $\ell_{t,k}$-flag of string $t$.
  Call a shift $\Delta\in[L/4]$ \emph{good} if $\ell_{t,K}\ge(1-24000\varepsilon)\ell_{t,[d]}$, i.e. that when weighted by flag length, the set $K$ comprises most of $[d]$.
  We know that since $\ell_{t,k}=\ceil{0.56\varepsilon^{-1}\ell_{s,k}}\in[0.56\varepsilon^{-1}\ell_{s,k},0.56\varepsilon^{-1}\ell_{s,k} (1+2\epsilon)]$ for all $k\in[d]$, we have that for good $\Delta$, \begin{align}
    \label{eq:blue-00-0}
    \ell_{s,K}
    \ge \frac{\frac{\epsilon}{0.56}\ell_{t,K}}{1+2\epsilon} 
    \ge \frac{\frac{\epsilon}{0.56}(1-24000\varepsilon)\ell_{t,[d]}}{1+2\epsilon}
    \ge \frac{\epsilon}{0.56} (1-24002\varepsilon)\ell_{t,[d]}
    \ge (1-24002\varepsilon)\ell_{s,[d]}.
  \end{align}
  Since string $t$ has type $m$-Blue-Yellow and $\ell_{t,k}> \ell_{s,k}\ge 2^m$ for every $k\in[d]$, we know that for every $k\in[d]$, at most $600\varepsilon L$ of the indices in $[L]$ are not Yellow $\ell_{t,k}$-flags in string $t$. 
  Hence, if $\Delta$ is chosen uniformly at random in $[L/4]$, for each $k\in[d]$, the probability that $k\not\in K$ is at most $600\varepsilon L / (L/4) = 2400\varepsilon$.
  By linearity of expectation,
  \[
    \E_{\Delta\sim[L/4]}[\ell_{t,[d]}-\ell_{t,K}] \le 2400\varepsilon \ell_{t,[d]},
  \]
  As $\ell_{t,K}\le \ell_{t,[d]}$ always, we have by Markov's inequality on $\ell_{t,[d]}-\ell_{t,K}$ that
  \[
    \Pr[\Delta\text{ is good}]
    = 1 - \Pr\left[ \ell_{t,[d]}-\ell_{t,K} \ge (24000\varepsilon)\ell_{t,[d]} \right] \ge 1- \frac{2400\varepsilon\ell_{t,[d]}}{24000\varepsilon\ell_{t,[d]}} = 0.9.
  \]

\smallskip\noindent\emph{Lower bounding expected LCS.}
Since $s_0''s_1's_1''s_2's_2''\dots s_d's_d''= s_{[1,(0.5+0.01\varepsilon)L]}$ and $t_0''t_1't_1''t_2't_2''\dots t_d't_d''$ is a prefix of $t_{[1+\Delta, (0.5+0.3\varepsilon)L+\Delta+1]}$, it suffices to show that, when $\Delta$ is good,
  \begin{align*}
    \LCS(s_0''s_1's_1''s_2's_2''\dots s_d's_d'', t_0''t_1't_1''t_2't_2''\dots t_d't_d'') \ge (0.5 + 0.24 \varepsilon)L.
  \end{align*}
  As the index $i_k$ is a Blue $\ell_{s,k}$-flag in $s$, substring $s_k'$ has at least $\eps^{-1}(\ell_{s,k}-1) \ge 0.5\varepsilon^{-1}\ell_{s,k}$ zeros.
  If $k\in K$, then $j_k$ is a Yellow $\ell_{t,k}$-flag in $t$, so substring $t_k'$ has at least $0.9(\ell_{t,k}-1) > 0.5\varepsilon^{-1}\ell_{s,k}$ zeros as well.
  Hence, for $k\in K$, we have
  \begin{align}
    \label{eq:blue-00-1}
    \LCS(s_k',t_k')\ge 0.5\varepsilon^{-1}\ell_{s,k}.
  \end{align}
  Furthermore, for $k\in [0,d]$, we have that $s_k''$ has $i_{k+1}  - i_k - \ell_{s,k}$ ones, and the number of ones in $t_k''$ is also
  \[
    j_{k+1}  - j_k -\ell_{t,k} 
    = (i_{k+1} +\Delta + \ell_{t,k} - \ell_{s,[k]}) - (i_k + \Delta +\ell_{t,[k-1]} - \ell_{s,[k-1]}) - \ell_{t,k}
    = i_{k+1} - i_k- \ell_{s,k},
  \]
  so for all $k\in [0,d]$,
  \[
    \LCS(s_k'',t_k'')\ge i_{k+1}-i_k-\ell_{s,k}.
  \]
  Indeed, we now have that for any good $\Delta$,
  \begin{align}
    \LCS(s_{[1,(0.5+0.01\varepsilon)L]}, t_{[1+\Delta, (0.5+0.3\varepsilon)L + \Delta]}) 
    \ &\ge \ 
    \sum_{k=1}^{d} \LCS(s_k',t_k')
    +\sum_{k=0}^{d} \LCS(s_k'',t_k'')\nonumber\\
    \ &\ge \  \sum_{k\in K} \LCS(s_k',t_k') + \sum_{k=0}^{d}( i_{k+1}-i_k-\ell_{s,k})\nonumber\\
    \ &\ge \   0.5\varepsilon^{-1}\ell_{s,K} + (i_{d+1}-i_0-\ell_{s,[d]}) \label{eq:blue-0-3}\\
    \ &\ge \  0.5\varepsilon^{-1}(1-24002\varepsilon)\ell_{s,[d]} + (0.5+0.01\varepsilon)L-2-\ell_{s,[d]}\label{eq:blue-0-4}\\
    \ &= \  0.5\varepsilon^{-1}(1-24004\varepsilon)\ell_{s,[d]} + (0.5+0.01\varepsilon)L-2\nonumber\\
    \ &\ge \  0.5\varepsilon^{-1}(1-24004\varepsilon)(0.5(\varepsilon^2-4\gamma) L) + (0.5+0.01\varepsilon)L \label{eq:blue-0-6}\\
    \ &\ge \ (0.5+0.24\varepsilon)L.\label{eq:blue-0-7}
  \end{align}
  In \eqref{eq:blue-0-3}, we used \eqref{eq:blue-00-1} and $\ell_{s,0}=0$.
  In \eqref{eq:blue-0-4}, we used \eqref{eq:blue-00-0} and that $i_{d+1}-i_0 = \floor{(0.5+0.01\varepsilon)L}-1 > (0.05+0.01\varepsilon)L-2$.
  In \eqref{eq:blue-0-6}, we used that $\ell_{s,[d]}=\sum_{k=1}^{d} b_s(i_k) \ge 0.5(\varepsilon^2-4\gamma)L$ by construction of $i_1,\dots,i_d$.
  In \eqref{eq:blue-0-7}, we used that $\varepsilon\le 10^{-6}$ and $\gamma\le 0.001\varepsilon^2$. This completes the proof as the probability that $\Delta \in [L/4]$ is good is at least $0.9$. 
\end{proof}

The next lemma shows that when two Blue-Yellow pairs occur consecutively (for technical reasons, we require that the second Blue-Yellow pair occurs on the reversed strings), we get a large common subsequence that also spans the same number of ones in both $s$ and $t$.
This allows us to piece together many such matchings and apply the Prefix/Suffix LCS Lemma (Lemma~\ref{lem:shift}), giving an overall gain in LCS over the 1/2 barrier.

\begin{lemma}[Blue-Yellow balanced matching lemma]
  \label{lem:blue2}
  Let $s=s_1s_2$ and $t=t_1t_2$ such that substrings $s_1,s_2,t_1,t_2$ each have $L$ ones and start with a one, and the pairs $(s_1,t_1)$ and $(\rev(t_2),\rev(s_2))$ are Blue-Yellow pairs.
  Then for $\Delta$ uniformly random in $[L/4]$, we have
  \[
    \E_{\Delta\sim[L/4]}\left[\LCS(\Drop_\Delta(s), \Drop_{-\Delta}(t)) + \Delta\right]\ge 2L + 0.12\varepsilon L.
  \] 

\end{lemma}
\begin{proof}
   We have $\Drop_\Delta(s)$ and $\Drop_{-\Delta}(t)$ each have $2L-\Delta$ ones, so we always have 
   \[
   \LCS(\Drop_\Delta(s), \Drop_{-\Delta}(t)) + \Delta \ge 2L.
   \]
   Thus, because $0.8 \cdot 0.16 \eps L > 0.12 \eps L$, it suffices to show that
   \[
    \Pr_{\Delta\sim[L/4]}\left[\LCS(\Drop_\Delta(s), \Drop_{-\Delta}(t))+\Delta\ge 2L + 0.16\varepsilon L\right]  \ge 0.8.
   \]
   Call $\Delta$ \emph{good} if both of the following inequalities are true:
   \begin{align}
    \LCS\left((s_1)_{[1,(0.5+0.01\varepsilon)L]}, (t_1)_{[1+\Delta,(0.5+0.3\varepsilon)L+\Delta]}\right) &\ge (0.5+0.24\varepsilon)L
    \label{eq:blue-1} \\
    \LCS\left(\rev(t_2)_{[1,(0.5+0.01\varepsilon)L]}, \rev(s_2)_{[1+\Delta,(0.5+0.3\varepsilon)L+\Delta]}\right) &\ge (0.5+0.24\varepsilon)L.
    \label{eq:blue-2}
   \end{align}
   
   By Lemma~\ref{lem:blue1}, first applied to the Blue-Yellow pair $(s_1,t_1)$, then applied to the Blue-Yellow pair $(\rev(t_2),\rev(s_2))$, each of \eqref{eq:blue-1} and \eqref{eq:blue-2} holds with probability at least $0.9$ for $\Delta$ sampled uniformly from $[L/4]$, so a random $\Delta$ in $[L/4]$ is good with probability at least $0.8$.

   Fix a good $\Delta$.
   We show that 
   \begin{equation}
    \label{eq:goal-lcs}   
     \LCS(\Drop_\Delta(s), \Drop_{-\Delta}(t))\ge 2L-\Delta + 0.16\varepsilon L.
    \end{equation}
   By the second part of Lemma~\ref{lem:rev},
   \begin{align*}
   \rev(t_2)_{[1,(0.5+0.01\varepsilon)L]} &=  \rev\left((t_2)_{[1+(0.5-0.01\varepsilon)L,L]}\right) \\
   \rev(s_2)_{[1+\Delta,(0.5+0.3\varepsilon)L+\Delta]} &=  \rev\left((s_2)_{[1+(0.5-0.3\varepsilon)L-\Delta,L-\Delta]}\right).
   \end{align*}
   Hence since $\rev(\cdot)$ preserves LCS (Lemma~\ref{lem:rev-lcs}), we have
   \begin{align}
   &\LCS\left((t_2)_{[1+(0.5-0.01\varepsilon)L,L]}, (s_2)_{[1+(0.5-0.3\varepsilon)L-\Delta,L-\Delta]}\right)  \nonumber\\
   &= \LCS\left(\rev(t_2)_{[1,(0.5+0.01\varepsilon)L]}, \rev(s_2)_{[1+\Delta,(0.5+0.3\varepsilon)L+\Delta]}\right) \nonumber\\
   &\ge (0.5+0.24\varepsilon)L \ . \nonumber
   \end{align}

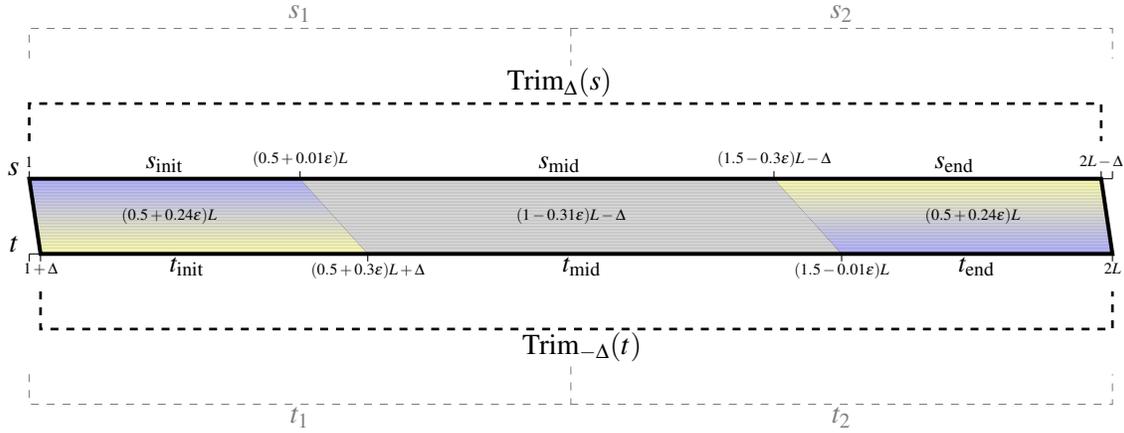
\begin{figure}[h]
\centering
\vspace{-5mm}

\scalebox{1}{
\begin{tikzpicture}
\drawpartbynoind{0}{0.5}{0}{15}{0}{0}{0}{12}
\drawpartgraynoind{0}{15}{0}{36}{0}{12}{0}{33}
\drawpartybnoind{0}{48}{0}{36}{0}{47.5}{0}{33}
\draw[line width=1.5](0,1) -- (47.5*\xspacing,1) -- (48*\xspacing,0) -- (0.5*\xspacing,0) -- cycle;
\node[label={[anchor=east]$s$}] at (0,1) {};
\node[label={[anchor=east]$t$}] at (0,0) {};
\draw (0,0) -- (48*\xspacing,0);
\draw (0,1) -- (48*\xspacing,1);

\draw[dashed,gray] (0,2.6) -- (0,3) -- (48*\xspacing,3) -- (48*\xspacing, 2.6);
\draw[dashed,gray] (0,-1.6) -- (0,-2) -- (48*\xspacing,-2) -- (48*\xspacing, -1.6);
\draw[dashed,gray] (24*\xspacing,3) -- (24*\xspacing,2.5);
\draw[dashed,gray] (24*\xspacing,-1.5) -- (24*\xspacing,-2);
\node[outer sep=-5, label={[anchor=south,gray]$s_{1}$}] at (12*\xspacing,3) {};
\node[outer sep=-5, label={[anchor=south,gray]$s_{2}$}] at (36*\xspacing,3) {};
\node[outer sep=-1,                           label={[anchor=north,gray]$t_{1}$}] at (12*\xspacing,-2) {};
\node[outer sep=-1,                           label={[anchor=north,gray]$t_{2}$}] at (36*\xspacing,-2) {};

\draw[dashed, line width=1] (0,1.5) -- (0,2) -- (47.5*\xspacing, 2) -- (47.5*\xspacing, 1.5);
\draw[dashed, line width=1] (0.5*\xspacing,-0.5) -- (0.5*\xspacing,-1) -- (48*\xspacing, -1) -- (48*\xspacing, -0.5);
\node[outer sep=-5, label={[anchor=south]$\Drop_\Delta(s)$}] at (23.5*\xspacing,2) {};
\node[outer sep=-1,                           label={[anchor=north]$\Drop_{-\Delta}(t)$}] at (24.5*\xspacing,-1) {};

\node[outer sep=-5, label={[anchor=south]\small$s_{\text{init}}$}] at (6*\xspacing,1) {};
\node[outer sep=-5, label={[anchor=south]\small$s_{\text{mid}}$}] at (23.5*\xspacing,1) {};
\node[outer sep=-5, label={[anchor=south]\small$s_{\text{end}}$}] at (41*\xspacing,1) {};
\node[outer sep=-1,                           label={[anchor=north]\small$t_{\text{init}}$}] at (7*\xspacing,0) {};
\node[outer sep=-1,                           label={[anchor=north]\small$t_{\text{mid}}$}] at (24.5*\xspacing,0) {};
\node[outer sep=-1,                           label={[anchor=north]\small$t_{\text{end}}$}] at (42*\xspacing,0) {};
\node[] at (6.2*\xspacing,0.5) {\tiny$(0.5+0.24\varepsilon)L$};
\node[] at (24*\xspacing,0.5) {\tiny$(1-0.31\varepsilon)L-\Delta$};
\node[] at (41.8*\xspacing,0.5) {\tiny$(0.5+0.24\varepsilon)L$};
\toplabel{0}{1};
\toplabel{12}{(0.5+0.01\varepsilon)L};
\toplabel{33}{(1.5-0.3\varepsilon)L-\Delta};
\toplabel{47.5}{2L-\Delta};
\toplabel{48}{};
\bottomlabel{0}{};
\bottomlabel{0.5}{1+\Delta};
\bottomlabel{15}{(0.5+0.3\varepsilon)L+\Delta};
\bottomlabel{36}{(1.5-0.01\varepsilon)L};
\bottomlabel{48}{2L};
\end{tikzpicture}
}
\vspace{-5mm}
\caption{Lemma~\ref{lem:blue2}. We obtain two ``trapezoids'' with good LCS obtained from applying Lemma~\ref{lem:blue1} to two Blue-Yellow pairs, $(s_1,t_1)$ and $(\rev(s_2),\rev(t_2))$, and match ones in between (in the grey region), giving an improved LCS for $\Drop_\Delta(s)$ and $\Drop_{-\Delta}(t)$.} 
\label{fig:blueyellow-case-2}
\end{figure}
Writing $\Drop_\Delta(s)=s_{\text{init}}s_{\text{mid}}s_{\text{end}}$ and $\Drop_{-\Delta}(t)=t_{\text{init}}t_{\text{mid}}t_{\text{end}}$ where (see Figure~\ref{fig:blueyellow-case-2})\footnote{A negligible detail: we do not need to shift the endpoints of substrings $s_{\text{mid}}$ and $t_{\text{mid}}$ by 1 because $L$ is a power of two so $\varepsilon L$ is never an integer}
   \begin{align}
     s_{\text{init}} &\defeq s_{[1,(0.5+0.01\varepsilon)L]} &&= (s_1)_{[1,(0.5+0.01\varepsilon)L]} \nonumber\\
     s_{\text{mid}} &\defeq s_{[(0.5+0.01\varepsilon)L, (1.5-0.3\varepsilon)L-\Delta+1]} \nonumber\\ 
     s_{\text{end}} &\defeq s_{[(1.5-0.3\varepsilon)L-\Delta+1,2L-\Delta]} &&= (s_2)_{[(0.5-0.3\varepsilon)L-\Delta+1,L-\Delta]} \nonumber\\
     t_{\text{init}} &\defeq t_{[1+\Delta,(0.5+0.3\varepsilon)L+\Delta]} &&= (t_1)_{[1+\Delta,(0.5+0.3\varepsilon)L+\Delta]}\nonumber\\
     t_{\text{mid}} &\defeq t_{[(0.5+0.3\varepsilon)L+\Delta, (1.5-0.01\varepsilon)L+1]}\nonumber\\
     t_{\text{end}} &\defeq t_{[(1.5-0.01\varepsilon)L+1,L]} &&= (t_2)_{[(0.5-0.01\varepsilon)L+1,L]} \ . \nonumber 
   \end{align}
   We have $s_{\text{mid}}$ and $t_{\text{mid}}$ both have at least $(1-0.31\varepsilon)L-\Delta$ ones.
   Hence, we have
   \begin{align*}
    \LCS(\Drop_\Delta(s),\Drop_{-\Delta}(t))
    \ &\ge \ \LCS(s_{\text{init}},t_{\text{init}}) + \LCS(s_{\text{mid}},t_{\text{mid}}) + \LCS(s_{\text{end}},t_{\text{end}})\nonumber\\
    \ &\ge \ (0.5+0.24\varepsilon)L +  (1-0.31\varepsilon)L-\Delta + (0.5+0.24\varepsilon)L  \\
    \ &> \  2L - \Delta + 0.16\varepsilon L,
   \end{align*}
  establishing \eqref{eq:goal-lcs} as desired.
\end{proof}

Combining Lemma~\ref{lem:blue2} with the Prefix/Suffix LCS Lemma (Lemma~\ref{lem:shift}) shows that when two strings have many Blue-Yellow pairs among their substrings at some scale, we get a large LCS.

\begin{lemma}[Many Blue-Yellow pairs implies good LCS]
\label{lem:blue3}
  Let $s$ and $t$ be strings with the same length and the same number $L=2^n$ of ones, and let $m^*\le n-10-\log\varepsilon^{-1}$.
  Suppose there exists a set $Z\subset[2^{n-m^*-1}]$ such that $|Z| > 2^{n-m^*-1}/10$ and for all $i\in Z$, the substring pairs $(s_{m^*,2i-1},t_{m^*,2i-1})$ and $(\rev(s_{m^*,2i}),\rev(t_{m^*,2i}))$ are Blue-Yellow pairs.
  Then,
  \begin{align}
    \LCS(s,t) > \left(0.5 + 0.0001\varepsilon \right)|s|.
  \end{align}
\end{lemma}
\begin{proof}
Let $L^* = 2^{m^*}$.
We may assume $L\ge (0.5-0.001\varepsilon)|s|$ or else we are done by having an LCS of $(0.5+0.001\varepsilon)|s|$ zeros.
For all $i\in Z$, substrings $s_{m^*+1,i}=s_{m^*,2i-1}s_{m^*,2i}$ and $t_{m^*+1,i}=t_{m^*,2i-1}t_{m^*,2i}$ satisfy the setup of Lemma~\ref{lem:blue2} with $L=L^*$, so we have, for all $i\in Z$,
  \begin{align*}
    \E_{\Delta\sim[L^*/4]}\big[\LCS\left(\Drop_{\Delta}(s_{m^*+1,i}), \Drop_{-\Delta}(t_{m^*+1,i})\right)+\Delta\big]
    \ge 2L^*+0.12\varepsilon L^* \ .
  \end{align*}
  By linearity of expectation, we have
  \begin{align*}
    \E_{\Delta\sim[L^*/4]}\left[\sum_{i\in Z}^{} \big[\LCS\left(\Drop_{\Delta}(s_{m^*+1,i}), \Drop_{-\Delta}(t_{m^*+1,i})\right)+\Delta\big]\right]
    \ge |Z|\cdot (2L^*+0.12\varepsilon L^*) \ .
  \end{align*}
  Hence, we may fix a $\Delta\in[L^*/4]$ for which 
  \begin{align*}
   \sum_{i\in Z}^{} \LCS\left(\Drop_{\Delta}(s_{m^*+1,i}), \Drop_{-\Delta}(t_{m^*+1,i})\right)
    \ge |Z|\cdot (2L^*-\Delta+0.12\varepsilon L^*) \ .
  \end{align*}
  Thus, the set $Z$ satisfies the setup of Lemma~\ref{lem:shift} with $n'=n$ and $m'=m^*+1$ and $L'=2L^*$ and $\delta' = 0.06\varepsilon$. 
  Hence, by Lemma~\ref{lem:shift},
 \[
    \LCS(s,t)
    \ge \left( 1 + \frac{0.06\varepsilon}{20} \right)L
    \ge \left( 1 + 0.003\varepsilon \right)(0.5-0.001\varepsilon)|s|
    > \left( 0.5 + 0.0001\varepsilon \right)|s| \ . \ \qedhere
    \]
\end{proof}

\section{Putting it all together}\label{sec:wrapup}

\subsection{Statistics}

We now prove our main theorem.
The first step is to define the \emph{statistics} of a string $w$. 
\begin{definition}[Statistics]
    \label{def:statistics}
  Let $w$ be a string with $L=2^n$ ones, and let $n_0\defeq \max(0,n - 200\gamma^{-3}\log n)$.
  Let the \emph{statistics} of string $w$ be a table of the following data:
  \begin{enumerate}
  \item For all $m\ge n_0$ and $i\ge 1$, the number of zeros and the number of ones in $w_{m,i}$ (the number of ones is always $2^m$).
  \item For all $m\ge n_0$ and $i\ge 1$, the type (see Definition~\ref{def:type}) of string $w_{m,i}$.
  \item The set $\mathcal{I}^{n_0}(w)$ of pairwise disjoint intervals $I\subset[L]$ that each have size $|I|\ge 2^{n_0}$ such that $w_I$ is imbalanced for each $I\in \mathcal{I}^{n_0}(w)$, and the sum $\sum_{I\in \mathcal{I}^{n_0}(w)}^{} |I|$ is maximized (if multiple such $\mathcal{I}^{n_0}(w)$ exist, break the tie arbitrarily).
  \item For each $I\in \mathcal{I}^{n_0}(w)$, the indices $x$ and $y$ such that substring $w_I$ starts at the $x$-th bit and ends at the $y$-th bit of $s$.
  \end{enumerate}
  We say two strings $s$ and $t$ {\it agree on statistics} if their tables of statistics are identical.
\end{definition}

This next lemma shows that it is possible to pigeonhole strings by their statistics, by showing that there are not too many possible statistics for a string.

\begin{lemma}
  There are at most $2^{\poly\log N}$ possible tables of statistics for a string of length $N$.
\label{lem:stats}
\end{lemma}
\begin{proof}
  A string $w$ of length $N$ has at most
  \[
  \sum_{m=n_0}^{n} 2^{n- m} = \poly\log N
  \]
  substrings $w_{m,i}$ that are considered in its table of statistics.
  Furthermore, for each substring $w_{m,i}$, there are at most $N+1$ choices for the number of zeros and the number of ones, and at most $O(N)$ choices for the type of $w_{m,i}$ (there is one Imbalanced type, $O(N)$ Green types, and $O(\log N)$ Blue-Yellow types), so there are at most $(\poly N)^{\#(\text{substrings }w_{m,i})}\le 2^{\poly\log N}$ choices for all the types and zero/one-counts in the table.
  Lastly, $\mathcal{I}^{n_0}(w)$ has at most $2^{n-n_0}=\log^{O_\varepsilon(1)}N$ intervals (because intervals have length at least $2^{n_0}$ and are disjoint), so there are at most $(2^{\poly\log N})^{|\mathcal{I}^{n_0}(w)|}\le 2^{\poly\log N}$ choices for $\mathcal{I}^{n_0}(w)$ and the locations of the endpoints of the intervals in $\mathcal{I}^{n_0}(w)$.
\end{proof}

\subsection{The Imbalanced case} 
This next lemma covers an easy case, when $s$ has many large imbalanced intervals.
In this case, our pigeonholing by statistics guarantees that $s$ and $t$ are imbalanced in common locations, allowing us to apply the imbalanced strategy to find a large LCS.
\begin{lemma}
\label{lem:main-0}
  Let $s$ and $t$ be strings that each start with a one, agree on statistics, and have $L=2^n$ ones where $n$ is sufficiently large.
  Suppose there exists a set $\mathcal{I}$ of pairwise disjoint intervals $I$ that each satisfy $|I|\ge 2^{n-200\gamma^{-3}\log n}$ such that $s_I$ is imbalanced for each $I \in \mathcal{I}$ and $\sum_{I\in \mathcal{I}} |I|\ge \frac{\varepsilon^5}{10}L$.
  Then
  \[
  \LCS(s,t)\ge \left(0.5+\frac{\varepsilon^{6}}{150}\right)|s|.
  \]
\end{lemma}
\begin{proof}
 Let $n_0= n-200\gamma^{-3}\log n$.
  We may assume that $L\ge \frac{|s|}{3}$, or else $s$ and $t$ each have $\frac{2|s|}{3}$ zeros and $\LCS(s,t) \ge \frac{2}{3}|s| > (0.5+\frac{\varepsilon^{6}}{150})|s|$.
  Since $s$ and $t$ agree on statistics, we have $\mathcal{I}^{n_0}(s)=\mathcal{I}^{n_0}(t)$.
  Furthermore, as $\mathcal{I}^{n_0}(s)$ maximizes the sum $\sum_{I\in \mathcal{I}^{n_0}(s)}^{} |I|$, we have $\sum_{I\in \mathcal{I}^{n_0}(s)}^{} |I|\ge \sum_{I\in\mathcal{I}}^{} |I|\ge \frac{\varepsilon^5}{10} L$.
  Let $\mathcal{I}^{n_0}(s) = \mathcal{I}^{n_0}(t)=\{I_1,\dots,I_k\}$, with $I_1<I_2<\cdots < I_k$ (these intervals are pairwise disjoint so the order is the obvious one).
  We thus may write
  \begin{align*}
    s = s_0''s_1's_1''s_2'\cdots s_k's_k'',\qquad
    t = t_0''t_1't_1''t_2'\cdots t_k't_k''
  \end{align*}
  where $s_j' = s_{I_j}$ for $j\in [1,k]$, and substring $s_j''$ consists of the bits between the end of substring $s_{j}'$ (or the beginning of string $s$ if $j=0$) and the beginning of substring $s_{j+1}'$ (or the end of string $s$ if $j=k$), and the partition of $t$ is defined analogously.
  By the definitions of $s_{I_j}$ and $t_{I_j}$, for $j\in [1,k]$, $s_j'$ and $t_j'$ have the same number $|I_j|$ of ones, and for $j\in [0,k]$, $s_j''$ and $t_j''$ have the same number of ones as well.
  Further, since $s$ and $t$ agree on statistics, for all $j\in[1,k]$, substrings $s_{I_j}$ and $t_{I_j}$ start and end in the same positions in their respective strings $s$ and $t$.
  In particular, $s_j'$ and $t_j'$ have the same length for all $j\in[1,k]$, and $s_j''$ and $t_j''$ have the same length for all $j\in[0,k]$.

  For each $j\in [0,k]$, the substrings $s_j''$ and $t_j''$ have the same length and the same number of ones, so $\LCS(s_j'',t_j'')\ge |s_j''|/2$.
  Additionally, for $j\in [1,k]$ the number of zeros in substrings $s_j'$ and $t_j'$ are equal and not in $(1\pm \varepsilon)|s_j'|$.
  Hence, by Lemma~\ref{lem:trivial-0}, we have $\LCS(s_j',t_j')\ge (1/2 + \varepsilon/5)|s_j'|$.
  Therefore we have
  \begin{align*}
    \LCS(s,t)
    &\ge \sum_{j=1}^{k} \LCS(s_j',t_j')
    + \sum_{j=0}^{k} \LCS(s_j'',t_j'') \ \ge \sum_{j=1}^{k} \left(\frac{1}{2} + \frac{\varepsilon}{5}\right)|s_j'|
    + \sum_{j=0}^{k} \frac{|s_j''|}{2} \\
    &= \frac{|s|}{2} + \frac{\varepsilon}{5}\sum_{j=1}^{k} |s_j'|   \ \ge \frac{|s|}{2} + \frac{\varepsilon^{6} L}{50} 
    \ \ge \left(\frac{1}{2} + \frac{\varepsilon^6}{150}\right)|s|,
  \end{align*}
  where the last two inequalities used that $|s_j'| \ge |I_j|$, $\sum_{j=1}^k |I_j| \ge \tfrac{\eps^5 L}{10}$, and $L \ge \tfrac{|s|}{3}$.
\end{proof}

\subsection{Combining the arguments for the Imbalanced, Green, and Blue-Yellow cases}
We now prove the main technical lemma, which shows that two strings that agree on statistics have LCS beating the $1/2$ barrier.
This establishes Theorem~\ref{thm:main-intro} up to a pigeonhole argument and an assumption about the number of ones being a power of two.
The proof consists of piecing together (1) the Imbalanced case, when $s$ and $t$ have many substrings of Imbalanced type, Lemma~\ref{lem:main-0}, (2) the Green case, when $s$ and $t$ have many substrings of Green type, and (3) the Blue-Yellow case, when $s$ and $t$ have many substrings of Blue-Yellow type.
These three cases correspond to the three matching strategies stated in the Overview (Section~\ref{sec:overview}).

\begin{lemma}
  \label{lem:main-1}
  There exists an absolute constant $\delta=\frac{\varepsilon^6}{150}$ such that the following holds for $n$ sufficiently large.
  Let $s$ and $t$ be strings that each start with a one, have $L=2^n$ ones each, such that $s$ and $t$ agree on statistics, and $\rev(s)$ and $\rev(t)$ agree on statistics.
  Then $\LCS(s,t)\ge (0.5+\delta)|s|$.
\end{lemma}
\begin{proof}
  First, suppose that at least $\gamma^2 L$ values of $i\in[L]$ satisfy $b_s(i)\ge 2^{n-200\gamma^{-3}\log n}$. 
  Define $1\le i_1<\dots<i_d$ such that $i_k$ is the smallest index such that $b_s(i_k)\ge 2^{n-200\gamma^{-3}\log n}$ and $i_k\ge i_{k-1}+b_s(i_{k-1})$ if $k>1$, and $d$ is the largest index such that $i_d$ is well-defined.
  By the definition of $b_s(i_k)$, index $i_k$ is a Blue $b_s(i_k)$-flag in $s$, so for the interval $I_k\defeq [i_k,\min(i_k+b_s(i_k)-1,L)]$, substring $s_{I_k}$ is imbalanced by Lemma~\ref{lem:blue-0}.
  Furthermore, since $i_k+b_s(i_k)\le i_{k+1}$ for $k=1,\dots,d-1$, we have $I_1,\dots,I_{d}$ are pairwise disjoint.
  Lastly, by minimality of each $i_k$, each index $i\in[L]$ with $b_s(i)\ge 2^{n-200\gamma^{-3}\log n}$ is in some interval $I_k$. 
  Thus, $\sum_{}^{} |I_k|\ge \gamma^2 L > \frac{\varepsilon^5}{10}$.
  Thus, we may apply Lemma~\ref{lem:main-0} to strings $s$ and $t$, giving $\LCS(s,t)\ge (0.5+\delta)|s|$.
  Hence, we may assume for the rest of the argument that
  \begin{itemize}
  \item At most $\gamma^2 L$ values of $i\in [L]$ satisfy $b_s(i)\ge 2^{n-200\gamma^{-3}\log n}$.
  \end{itemize}
  Similarly, we may assume 
  \begin{itemize}
  \item At most $\gamma^2 L$ values of $i\in [L]$ satisfy $b_{\rev(s)}(i) \ge 2^{n-200\gamma^{-3}\log n}$ (applying Lemma~\ref{lem:main-0} to $\rev(s)$ and $\rev(t)$),
  \item At most $\gamma^2 L$ values of $i\in [L]$ satisfy $b_{t}(i) \ge 2^{n-200\gamma^{-3}\log n}$ (applying Lemma~\ref{lem:main-0} to $t$ and $s$), and
  \item At most $\gamma^2 L$ values of $i\in [L]$ satisfy $b_{\rev(t)}(i) \ge 2^{n-200\gamma^{-3}\log n}$ (applying Lemma~\ref{lem:main-0} to $\rev(t)$ and $\rev(s)$). 
  \end{itemize}

  Hence, we may apply the Substring Blue-flag-balance Lemma, Lemma~\ref{lem:entropy}, to each of the strings $s$, $\rev(s)$, $t$, and $\rev(t)$ with $\beta = \gamma$.
  This shows the existence of some scale $m^*$ with $n-10-\log\delta^{-1}\ge m^*\ge n-150\gamma^{-3}\log n$ such that the following hold:
  \begin{itemize}
  \item For at least $(1-3\gamma) \cdot 2^{n-m^*}$ values of $i\in[2^{n-m^*}]$, the substring $s_{m^*,i}$ is $6\gamma$-Blue-flag-balanced,
  \item For at least $(1-3\gamma) \cdot 2^{n-m^*}$ values of $i\in[2^{n-m^*}]$, the substring $\rev(s_{m^*,i})$ is $6\gamma$-Blue-flag-balanced,
  \item For at least $(1-3\gamma) \cdot 2^{n-m^*}$ values of $i\in[2^{n-m^*}]$, the substring $t_{m^*,i}$ is $6\gamma$-Blue-flag-balanced, and
  \item For at least $(1-3\gamma) \cdot 2^{n-m^*}$ values of $i\in[2^{n-m^*}]$, the substring $\rev(t_{m^*,i})$ is $6\gamma$-Blue-flag-balanced.
  \end{itemize}
  Indeed, for each bullet, Lemma~\ref{lem:entropy} implies there are at most $32\gamma^{-3}\log n$ values of $m^*$ for which the condition does not hold, so among $150\gamma^{-3}\log n - 10-\log\delta^{-1} > 128\gamma^{-3}\log n$ values of $m^*$, some $m^*$ allows all four conditions to hold.
  Fix this $m^*$ and let
  \[ 
  L^*\defeq 2^{m^*} \ . 
  \]
    Since the number of ones in $s$ and $t$ are a power of two, $s$ and $t$ agree on statistics, and $m^* > n_0$, where $n_0$ is from Definition~\ref{def:statistics}, we have $s=s_{m^*,1}\cdots s_{m^*,2^{n-m^*}}$ and $t=t_{m^*,1},\dots,t_{m^*,2^{n-m^*}}$, where, for all $i=1,\dots,2^{n-m^*}$, substrings $s_{m^*,i}$ and $t_{m^*,i}$ have the same number $L^*$ of ones and the also the same number of zeros, and thus have the same length.
  Similarly, we may write $\rev(s)=\rev(s_{m^*,2^{n-m^*}})\cdots \rev(s_{m^*,1})$ and $\rev(t)=\rev(t_{m^*,2^{n-m^*}})\cdots \rev(t_{m^*,1})$.
  
  Since $n$ is sufficiently large, for all $i\in[2^{n-m^*}]$, substrings $s_{m^*,i}$ and $\rev(s_{m^*,i})$ have types by Lemma~\ref{lem:type}.
  Let $Z_1$ (resp. $\bar Z_1$) be the set of $i\in[2^{n-m^*}]$ such that substring $s_{m^*,i}$ (resp. $\rev(s_{m^*,i})$) is Imbalanced.
  Let $Z_2$ (resp. $\bar Z_2$) be the set of $i\in[2^{n-m^*}]$ such that substring $s_{m^*,i}$ (resp. $\rev(s_{m^*,i})$) is $\ell$-Green for some $\ell$.
  Let $Z_3$ (resp. $\bar Z_3$) be the set of $i\in[2^{n-m^*}]$ such that substring $s_{m^*,i}$ (resp. $\rev(s_{m^*,i})$) is $m$-Blue-Yellow for some $m$.
  Since $|Z_1|+|Z_2|+|Z_3|=|\bar Z_1|+|\bar Z_2|+|\bar Z_3|= 2^{n-m^*}$, we have the following cases covering all possibilities.

\smallskip
  \textbf{Case 1a.} $|Z_1|\ge 2^{n-m^*}/10$. 
  In this case, for each $i\in Z_1$, because $s_{m^*,i}$ is type Imbalanced, there exists some interval $J_i\subset I_{m^*,i}$ with $|J_i|\ge \varepsilon^{5} L^* > 2^{n-200\gamma^{-3}\log n}$ such that $s_{J_i}$ is imbalanced. 
  Since the intervals $I_{m^*,i}$ are pairwise disjoint, the intervals $J_i$ are pairwise disjoint.
  Then setting $\mathcal{I}' = \{J_i:i\in Z_1\}$, we have that $\sum_{J\in\mathcal{I}'}^{} |J|\ge \varepsilon^{5} L^* \cdot |Z_1| = \frac{\varepsilon^{5}}{10}\cdot L$. 
  Hence, by Lemma~\ref{lem:main-0}, we have $\LCS(s,t)\ge (0.5+\frac{\varepsilon^6}{150})|s| = (0.5+\delta)|s|$.

\smallskip
  \textbf{Case 1b.} $|\bar Z_1|\ge 2^{n-m^*}/10$.
  By an identical argument to Case 1a, we can show $\LCS(\rev(s),\rev(t))\ge (0.5+\delta)|s|$, which implies that $\LCS(s,t)\ge (0.5+\delta)|s|$.

\medskip
  \textbf{Case 2a.} $|Z_2|\ge 2^{n-m^*}/10$.
  Since $m^*\le n-10-\log\varepsilon^{-5}$ by definition of $m^*$, we may apply Lemma~\ref{lem:green2} to strings $s$ and $t$ with subset $Z\subset[2^{n-m^*}]$.
  By Lemma~\ref{lem:green2}, we have $\LCS(s,t)\ge (0.5+\frac{\varepsilon^5}{5000})|s| > (0.5+\delta)|s|$. 

\smallskip
  \textbf{Case 2b.} $|\bar Z_2|\ge 2^{n-m^*}/10$.
  By an identical argument to Case 2a, we can show $\LCS(\rev(s),\rev(t))\ge (0.5+\delta)|s|$, which implies that $\LCS(s,t)\ge (0.5+\delta)|s|$.

\medskip
  \textbf{Case 3.} $|Z_3|\ge \frac{4}{5}\cdot 2^{n-m^*}$ and $|\bar Z_3|\ge \frac{4}{5}\cdot 2^{n-m^*}$.
  Let $Z_3'$ be the set of $i\in[2^{n-m^*-1}]$ such that the following hold:
  \begin{itemize}
  \item Substring $s_{m^*,2i-1}$ is $6\gamma$-Blue-flag-balanced.
  \item Substring $t_{m^*,2i-1}$ is $6\gamma$-Blue-flag-balanced.
  \item Substring $\rev(s_{m^*,2i})$ is $6\gamma$-Blue-flag-balanced.
  \item Substring $\rev(t_{m^*,2i})$ is $6\gamma$-Blue-flag-balanced.
  \item We have $2i-1\in Z_3$.
  \item We have $2i\in \bar Z_3$.
  \end{itemize}
  By choice of $m^*$, the first four conditions above each fail for at most $3\gamma \cdot 2^{n-m^*}$ values of $i\in [2^{n-m^*-1}]$.
  Since $|Z_3|\ge \frac{4}{5}\cdot 2^{n-m^*}$, the fifth condition fails for at most $\frac{1}{5}\cdot 2^{n-m^*}$ values of $i$, and similarly the last condition fails for at most $\frac{1}{5}\cdot 2^{n-m^*}$ values of $i$.
  Since there are $\frac{1}{2}\cdot 2^{n-m^*}$ values of $i\in[2^{n-m^*-1}]$, we have that $Z_3'$ has size at least $(\frac{1}{2}-12\gamma - \frac{2}{5}) \cdot 2^{n-m^*} > \frac{1}{10}\cdot 2^{n-m^*-1}$ (recall $\gamma=10^{-15}$ is very small).

  Fix $i\in Z_3'$.  We claim that strings $s$ and $t$ with parameter $m^*$ and set $Z$ satisfy the setup of Lemma~\ref{lem:blue3}.
  The bound on $m^*\le n-10-\log\varepsilon^{-1}$ follows from the definition of $m^*$.
  We thus need to show that, for all $i\in Z_3'$, the pairs $(s_{m^*,2i-1},t_{m^*,2i-1})$ and $(\rev(s_{m^*,2i}),\rev(t_{m^*,2i}))$ are a Blue-Yellow pairs.
  Since $2i-1\in Z_3$, there exists some integer $m$ such that substring $s_{m^*,2i-1}$ is type $m$-Blue-Yellow, and since $s$ and $t$ agree on statistics, substring $t_{m^*,2i-1}$ is also type $m$-Blue-Yellow.
  Since $i\in Z_3'$, we have that strings $s_{m^*,2i-1}$ and $t_{m^*,2i-1}$ are both $6\gamma$-Blue-flag-balanced, so $(s_{m^*,2i-1},t_{m^*,2i-1})$ form a Blue-Yellow pair.
  Similarly, since $2i\in \bar Z_3$ and $\rev(s)$ and $\rev(t)$ agree on statistics, there exists some integer $m'$ such that substrings $\rev(s_{m^*,2i}) = \rev(s)_{m^*,2^{n-m^*}-(2i-1)}$ and $\rev(t_{m^*,2i})$ have type $m'$-Blue-Yellow.
  Since $i\in Z_3'$, we have that strings $\rev(s_{m^*,2i})$ and $\rev(t_{m^*,2i})$ are both $6\gamma$-Blue-flag-balanced, so $(\rev(s_{m^*,2i}),\rev(t_{m^*,2i}))$ form a Blue-Yellow pair, as desired.
  Thus, by Lemma~\ref{lem:blue3}, we have
  \[
    \LCS(s,t)
    \ge \left( 0.5 + 0.0001\varepsilon \right)|s|
    > \left( 0.5+\delta \right)|s|.
  \]
  In all cases we have shown that $\LCS(s,t)\ge (0.5+\delta)|s|$, proving the lemma.
\end{proof}

\subsection{Finishing the proof}

To prove the main theorem, we need to find two strings that agree on statistics, and remove the assumption that the number of ones is a power of two.
\begin{theorem}
  \label{thm:main}
  There exists absolute constants $A>0$ and $\delta=\frac{\varepsilon^6}{900}$ such that the following holds for $N$ sufficiently large.
  Let $C\subset\{0,1\}^N$ be a code with at least $2^{\log^A N}$ strings. 
  Then $C$ contains two strings $s$ and $t$ such that $\LCS(s,t)\ge (0.5+\delta)N$. 
\end{theorem}
\begin{proof}
  By Lemma~\ref{lem:stats}, there exists a constant $A'$ such that the number of possible tables of statistics for a string of length $N$ is at most $2^{O(\log^{A'}N)}$.
  We pick $A=A'+1$.
  By removing at most half of the elements of $C$, we may assume that every string in $C$ starts with the same bit, and without loss of generality we may assume that every string starts with a one.
  Let $2^n$ be the largest power of two less than $N/3$, so that $2^n\ge N/6$.
  If there exist two strings $s$ and $t$ in $C$ with less than $2^n$ ones, then we have $\LCS(s,t)\ge 2N/3$, so we may assume that all but at most one string in $C$ has at least $2^n$ ones.

  By the pigeonhole principle, there exist two strings $s,t \in C$ such that
  \begin{itemize}
  \item $s$ and $t$ have the same number $L\ge 2^n$ of ones,
  \item $s_{n,1}$ and $t_{n,1}$ agree on statistics, and
  \item $\rev(s_{n,1})$ and $\rev(t_{n,1})$ agree on statistics.
  \end{itemize}
  Indeed, the total number of tables of statistics for each of substrings $s_{n,1}$ and $\rev(s_{n,1})$ is at most $2^{O(\log^{A'}N)}$, so the total number of pigeonholes here is at most $2^{O(\log^{A'}N)} < 2^{\log^A N} = |C|$ if $N$ is sufficiently large.

  Let $s=s_{n,1}s'$ and $t=t_{n,1}t'$.
  Since strings $s$ and $t$ have the same length and same number of ones, and substrings $s_{n,1}$ and $t_{n,1}$ agree on statistics, we have that the suffixes $s'$ and $t'$ have the same length and the same number of ones as well.
  Thus, $\LCS(s',t')\ge |s'|/2$.
  By applying Lemma~\ref{lem:main-1} to strings $s_{n,1}$ and $t_{n,1}$, we have
  \[
  \LCS(s_{n,1},t_{n,1})\ge \left(0.5+\frac{\varepsilon^6}{150}\right)|s_{n,1}|.
  \]
  Hence, we have
  \[
    \LCS(s,t) \ge \LCS(s_{n,1},t_{n,1}) + \LCS(s',t') \ge \left(0.5+\frac{\varepsilon^6}{150}\right)|s_{n,1}| + 0.5|s'| = 0.5N + \frac{\varepsilon^6}{150} |s_{n,1}| \ge (0.5+\delta)N,
  \]
  as desired. In the last inequality, we used that $|s_{n,1}|\ge 2^n \ge N/6$.
\end{proof}

\section{Conclusion and open questions}\label{sec:concluding}

We gave the first non-trivial upper bound the zero-rate threshold for bit-deletions, showing that it is strictly less than $\sfrac12$. We achieved this via a structural lemma that classifies the oscillation patterns of $0$'s and $1$'s in a balanced string, and then exploiting it to carefully orchestrate a noticeable advantage over just matching $1$'s by switching to matching $0$'s at judicious points.

Now that we finally have a resolution to the first order question of whether $p^{\thr}_{\del}=\sfrac12$, it opens up the opportunity to address several related questions. We list some salient directions for future work below.

\begin{enumerate}
    \item An obvious and major challenge is to determine the exact value of the zero-rate threshold for bit-deletions which remains unknown. Our current state of knowledge is $\sqrt{2}-1 \le p^{\thr}_{\del} \le 1/2-\delta_0$, for $\delta_0 > 0$ given by Theorem~\ref{thm:main-intro}. The true value of $\sfrac12-p^{\thr}_{\del}$ is presumably much bigger than the minuscule $\delta_0$ our proof yields.
We made no attempt to optimize $\delta_0$, but it is likely to be very small regardless. 
We hazard a guess that the true value might be closer to the lower bound. More audaciously, one might even postulate that $p^{\thr}_{\del}=\sqrt{2}-1$ since this threshold appears to be the limit of the techniques in the spirit of \cite{BGH17}.

\item As an interesting and necessary step toward improving the upper bound on $p^{\thr}_{\del}$, can one obtain better upper bounds on the \emph{span} of a large enough code (which are implied by, but possibly easier to establish than, a lower bound on LCS)? Here, we define the \emph{span} of two strings $s$ and $t$ as the minimum ratio $(|s'| + |t'|)/\LCS(s',t')$ over all pairs of (contiguous) substrings $s'\subseteq s$ and $t'\subseteq t$ which are of lengths $|s'|\ge \Omega(|s|)$, $|t'|\ge \Omega(|t|)$, and the span of a code is the minimum span between any two distinct strings in the code.
 Note that if $C$ has span $\alpha$, then the LCS of any two distinct strings in $C$ is at most $2N/\alpha$. The codes of \cite{BGH17} are based on the construction of codes (a variant of the Bukh-Ma code) of growing size with  span at least $2+\sqrt{2}$.
Prior to our work, no non-trivial upper bound (bounded away from the trivial limit of $4$) was known on the span of positive-rate codes.

We point out that using our techniques, proving that there exist two codewords with span at most $4-\delta_0$ is easier than proving there exist two codewords with LCS at least $(1/2+\delta)N$.
To show span at most $4-\delta_0$, matched flags do not have to be at similar locations in the two strings, so we have more flexibility with our random shifting argument. In particular, we can apply the structure lemma to entire codewords rather than dyadic substrings as we do in Lemma~\ref{lem:main-1}, so we do not need to  combine LCS in prefixes/suffixes with Lemma~\ref{lem:shift}. Furthermore, in the Blue-Yellow case, it is okay to use an imbalanced matching where say, we match Blue flags only in $s$, consuming more ones in $t$, and hence we do not need the string regularity and string reversal $\rev(\cdot)$ arguments which were used to ensure a balanced matching. 

\item Our quasi-polynomial in $N$ upper bound on the size of codes $C \subset \{0,1\}^N$ with $\LCS(C) \le (\sfrac12+\delta_0)N$ can likely be improved with some more care in the argument, though we settled for it for sake of simplicity. 
We conjecture that in fact $|C| \le O(\log N)$.
This would be tight, as the Bukh-Ma code has size $\Omega(\log N)$ and $\LCS(C)\le (\sfrac12+\delta_0)N$.
As evidence towards this conjecture, our techniques can be used to show that $|C|\le O(\log N)$ when any two codewords have span at most $4-\delta_0$, by following the sketch in the previous item.

\item  For the $q$-ary alphabet, the $q$-ary codes in \cite{BGH17} correct a fraction of deletions approaching $1 - \tfrac{2}{q+\sqrt{q}}$, and thus we have 
\begin{equation}
\label{eq:q-ary-bounds}
1- \tfrac{2}{q+\sqrt{q}} \le p^{\thr}_{\del}(q) \le 1-\tfrac{1+2\delta_0}{q} \ .
\end{equation}
An interesting question is to determine the infimum of constants $c$ such that $p^{\thr}_{\del}(q) \ge 1-c/q$ for large enough $q$. The bounds in \eqref{eq:q-ary-bounds} imply that $c \in [1+2\delta_0,2]$.

\item In the list-decoding model, there are codes of rate $\poly(\eps)$ which can be list-decoded from a fraction $1/2-\eps$ of deletions with list-size $\poly(1/\eps)$. Can one prove a lower bound $L(\eps)$ on the list-size for list decoding from $(1/2-\eps)$ fraction of deletions such that $L(\eps) \to \infty$ as $\eps \to 0$? What is the optimal growth rate of the list-size as a function of $\eps$? (For correcting bit-flips, the optimal list-size is known to be $\Theta(1/\eps^2)$~\cite{blinovsky,GV10}.)

\item A {\it pair of twins} in a string $s$ is a pair of two disjoint subsequences (recall that in our language {\it subsequences} of strings are not necessarily contiguous, unlike substrings) of $s$ which are identical. A natural question to consider is: what is the length $t(N)$ of the longest pair of twins guaranteed to exist in any $s\in \{0,1\}^N$? This question is relevant for two reasons. First, the question of finding twins within a single string is closely related to the problem of finding longest common subsequences between distinct strings. Second, the twins problem was solved asymptotically by~\cite{APP13} using the regularity technique, which is also one of the ingredients in our work. It is now known that
\[
\frac{N}{2} - O\Big( \frac{N (\log \log N)^{1/4}}{(\log N)^{1/4}}\Big) \le t(N) \le \frac{N}{2} - \Omega(\log N),
\]
and it would be interesting to determine the exact growth rate of the lower-order term.
\end{enumerate}

\noindent {\bf Acknowledgments.} The authors are grateful to Jacob Fox for helpful conversations and for bringing \cite{APP13} to our attention.

\bibliographystyle{alpha}
\bibliography{deletions}

\end{document}